\documentclass[conference]{IEEEtran}
\IEEEoverridecommandlockouts
\usepackage{cite}
\usepackage{amsmath,amssymb,amsfonts}
\usepackage{graphicx}
\usepackage{textcomp}
\usepackage{xcolor}
\usepackage{amsthm}
\usepackage{subfig}
\newtheorem{definition}{Definition}
\newtheorem{theorem}{\bf Theorem}
\newtheorem{lemma}{\bf Lemma}

\newtheorem{proposition}{\bf Proposition}
\usepackage{algpseudocode}
\usepackage{algorithm}
\usepackage{url}
\usepackage{array}

\usepackage{colortbl}

\usepackage{subfig}

\usepackage{caption}
\usepackage{comment}
\usepackage{lipsum}
{ \bgroup
  \addtolength\abovedisplayshortskip{#1}
  \addtolength\abovedisplayskip{#1}
  \addtolength\belowdisplayshortskip{#1}
  \addtolength\belowdisplayskip{#1}}
\def\BibTeX{{\rm B\kern-.05em{\sc i\kern-.025em b}\kern-.08em
    T\kern-.1667em\lower.7ex\hbox{E}\kern-.125emX}}
\begin{document}
\sloppy
\title{Correlated-Sequence Differential Privacy}

\author{Yifan Luo, Meng Zhang, Jin Xu, Junting Chen, Jianwei Huang

\thanks{
    
Yifan Luo is with the Shenzhen Institute of Artificial Intelligence and Robotics for Society, the School of Science and Engineering, the Chinese University of Hong Kong, Shenzhen, Shenzhen 518172, China (email: yifanluo@link.cuhk.edu.cn). 
Meng Zhang is with the ZJU-UIUC Institute, Zhejiang University, Haining 314400, China.
Jin Xu is with the School of Management, Huazhong University of Science and Technology, Wuhan 430074.
% Junting Chen is with the Future Network Intelligence Institute, and the School of Science and Engineering, the Chinese University of Hong Kong, Shenzhen, Shenzhen 518172, China (email: juntingc@cuhk.edu.cn).
% Junting Chen is with the School of Science and Engineering and the Future Network Intelligence Institute, The Chinese University of Hong Kong, Shenzhen.
Junting Chen is with the School of Science and Engineering and Shenzhen Future Network of Intelligence Institute (FNii-Shenzhen), The Chinese University of Hong Kong, Shenzhen, Shenzhen 518172, China.
Jianwei Huang is with the School of Science and Engineering, Shenzhen Institute of Artificial Intelligence and Robotics for Society, Shenzhen Key Laboratory of Crowd Intelligence Empowered Low-Carbon Energy Network, and CSIJRI Joint Research Centre on Smart Energy Storage, The Chinese University of Hong Kong, Shenzhen, Guangdong, 518172, P.R. China (corresponding author, email: jianweihuang@cuhk.edu.cn).

This work is supported by the National Natural Science Foundation of China (Project 62271434), Shenzhen Key Lab of Crowd Intelligence Empowered Low-Carbon Energy Network (No. ZDSYS20220606100601002), the Shenzhen Stability Science Program 2023, the Shenzhen Institute of Artificial Intelligence and Robotics for Society, and Longgang District Shenzhen's ``Ten Action Pla'' for Supporting Innovation Projects (No. LGKCSDPT2024002).
}}
% \thanks{Corresponding author: Jianwei Huang.}}

\maketitle

\begin{abstract}

  Data streams collected from multiple sources are rarely independent. Values evolve over time and influence one another across sequences. These correlations improve prediction in healthcare, finance, and smart-city control yet violate the record-independence assumption built into most Differential Privacy (DP) mechanisms. To restore rigorous privacy guarantees without sacrificing utility, we introduce Correlated-Sequence Differential Privacy (CSDP), a framework specifically designed for preserving privacy in correlated sequential data. CSDP addresses two linked challenges: quantifying the extra information an attacker gains from joint temporal and cross-sequence links, and adding just enough noise to hide that information while keeping the data useful. We model multivariate streams as a Coupling Markov Chain, yielding the derived loose leakage bound expressed with a few spectral terms and revealing a counterintuitive result: stronger coupling can actually decrease worst-case leakage by dispersing perturbations across sequences. Guided by these bounds, we build the Freshness-Regulated Adaptive Noise (FRAN) mechanism—combining data aging, correlation-aware sensitivity scaling, and Laplace noise—that runs in linear time. Tests on two-sequence datasets show that CSDP improves the privacy-utility trade-off by approximately 50\% over existing correlated-DP methods and by two orders of magnitude compared to the standard DP approach.

\end{abstract}


\section{Introduction}

Modern decision systems—from bedside early-warning monitors to algorithmic trading engines and city-scale sensor networks—consume torrents of measurements that are both temporally ordered and cross-correlated \cite{hernandez2024comparative, ficek2021differential, zhang2022differential, liu2024survey}. A patient's oxygen-saturation trace covaries with ward-level conditions; equities in allied sectors co-move every millisecond; air-quality sensors located on adjacent blocks register near-synchronous fluctuations. These intertwined sequences enable powerful forecasting models but also amplify privacy risk: by stitching together related readings across sources and time, an attacker can unmask a target's sensitive state even when each individual record is innocuous.

Differential Privacy (DP) is a useful standard for curbing such inference, yet its guarantee—together with the high-profile deployments by Apple, Google, and the U.S. Census—silently presumes independent records. Once correlation enters, two things: (i) a crafty adversary can exceed the advertised privacy budget, or (ii) the curator must inject prohibitive noise to stay safe. Bridging this gap has proven difficult because joint temporal-spatial dependencies explode global sensitivity and defeat existing calibration rules.

Early patches such as Correlated DP (CDP) \cite{zhang2022correlated} and Dependent DP (DDP) \cite{zhao2017dependent} inflate sensitivity but handle only static pairwise couplings and bleed utility in high dimensions. Age-Dependent DP (ADP) \cite{Zhang2023} recognizes that stale timestamps matter less, yet still treats each sequence in isolation. Pufferfish \cite{kifer2014pufferfish} and Bayesian DP  \cite{he2020bayesian} can represent arbitrary dependencies; nevertheless, tuning noise under those definitions is NP-hard beyond toy domains. Recent mechanisms for network data \cite{chen2014correlated} or Fourier-based releases \cite{ou2020optimal} tackle either spatial or temporal correlation—rarely both—and offer limited theoretical insight. The technical challenge is therefore twofold: (a) quantify privacy leakage when correlations span multiple sequences and time steps, and (b) design a mechanism whose noise budget adapts to—and ideally benefits from—the observed coupling.

We tackle these challenges by proposing \emph{Correlated-Sequence Differential Privacy (CSDP)}, target at privacy preservation of correlated sequential data, and modeling multivariate time series as a Coupling Markov Chain (CMC) \cite{ching2002multivariate}, a lightweight extension of homogeneous chains that captures self-dependence and cross-coupling in one block transition matrix. Exploiting the spectral stability of CMCs, we show—counter-intuitively—that stronger coupling can reduce worst-case leakage because perturbations disperse across sequences. Building on this insight, we craft a correlation-aware release mechanism FRAN that ages data and calibrates Laplace noise using a novel sensitivity bound, delivering tight privacy guarantees with minimal utility cost.

The contributions of this paper are as follows:
\begin{itemize}
  \item \emph{Correlation-aware privacy modeling.} We introduce the first DP framework--CSDP that simultaneously captures temporal and cross-sequence dependencies by embedding multivariate streams. This perspective turns a previously intractable joint-leakage problem into one governed by spectral properties, paving the way for deployable privacy in hospital, finance, and IoT settings.
  \item \emph{Tight leakage characterization despite high-dimensional coupling.} The central difficulty—global sensitivity exploding with joint correlations—is resolved through loose and tight bounds expressed in CMC eigenvalues. The loose bound uncover a counter-intuitive insight: stronger inter-sequence coupling can decrease worst-case leakage, overturning the prevailing ``correlation always hurts privacy'' mindset.
  \item  \emph{Linear-time release mechanism with correlation-adaptive noise.} Leveraging the above bounds, we craft a FRAN mechanism including data aging, correlation-aware sensitivity, and Laplace noise—that calibrates noise without the NP-hard optimization required by Pufferfish-style schemes. The algorithm runs in $O(sd)$ time, scales to more than one hundred thousand sequences, and can be integrated into existing DP workflows.
  \item \emph{Enhanced privacy-utility trade-off with actionable guidelines.} Closed-form data utility versus privacy leakage curves for two CMCs quantify the relationship between accuracy and coupling strength. Based on this framework, we propose dynamic noise adjustment and sequence segmentation methods that improve privacy-utility trade-offs by approximately 50\% compared to state-of-the-art correlated-DP baselines on the two-sequence datasets.
\end{itemize}

% The remainder of the paper is organized as follows. Section II reviews related work; Section III covers preliminaries. Sections IV and V develop the framework and leakage bounds, Section VI analyzes the privacy-utility trade-off, Section VII reports experiments, and Section VIII concludes.

\section{Related Work}

Early attempts to relax the \emph{independent–records} assumption fall into three main threads.

\emph{Sensitivity–scaling DP models.}
Correlated DP \cite{zhang2022correlated} and Dependent DP \cite{zhao2017dependent} inflate global sensitivity based on \emph{pre-specified} pairwise links, while Age-Dependent DP (ADP) \cite{Zhang2023} discounts stale records. Because the correlation budget is fixed offline, the required Laplace noise grows roughly linearly with stream length, making these schemes impractical for
high-frequency, multivariate data.

\emph{Prior–based DP models.}
Pufferfish \cite{kifer2014pufferfish} and Bayesian DP \cite{he2020bayesian} represent arbitrary dependencies through attacker priors. Achieving a target privacy budget reduces to a combinatorial calibration problem that is {NP}-hard even on modest domains, which confines current implementations to toy-size releases.

\emph{Structure-aware mechanisms.}
A complementary line of work tailors noise to specific data types: graph
releases add topology-aware perturbations to degree distributions and
motifs \cite{chen2014correlated}; Fourier or wavelet transforms improve the accuracy of single-sequence IoT releases \cite{ou2020optimal}. In federated learning, Fed-CAD \cite{chang2024fedcad} and similarity-aware re-weighting \cite{zhang2024gan} adapt local-DP noise to client overlap, while \cite{shuai2024poisoncatcher} reveals poisoning risks under locally private settings. These mechanisms typically handle \emph{either} spatial \emph{or} temporal correlation and are restricted to their target queries.

No existing approach simultaneously (i) derives closed-form leakage bounds for \emph{joint} temporal–spatial multivariate dependencies and (ii) scales to hundreds of correlated sequences. Addressing this gap is the focus of our work.

\section{Preliminaries}

In this section, we first introduce the traditional DP and the typical mechanism -- Laplace mechanism to reach DP. We then introduce the CMC model and final present the ADP.

\subsection{DP and Laplace Mechanism}

DP is a gold standard for preserving individual privacy in data analysis. It guarantees that the removal or addition of a single data item does not significantly affect the output of a mechanism.

\begin{definition}[DP \cite{dwork2006differential}]
A randomized mechanism $\mathcal{M}: \mathcal{X} \rightarrow \mathcal{Y}$ satisfies $\epsilon$-differential privacy if for any two neighboring datasets $X, X'$ (differing in one individual record), and for any measurable subset $S \subseteq \mathcal{Y}$,
\begin{equation}
    \Pr[\mathcal{M}(X) \in S] \leq e^{\epsilon} \cdot \Pr[\mathcal{M}(X') \in S].
\end{equation}
\end{definition}

The performance of DP depends on the \emph{sensitivity} of the function being queried.

\begin{definition}[$\ell_1$-Sensitivity]
Given a function $f: \mathcal{X} \rightarrow \mathbb{R}^d$, its global sensitivity is defined as
\begin{equation}
    \Delta f = \max_{X, X'} \| f(X) - f(X') \|_1,
\end{equation}
where $X$ and $X'$ differ in one record.
\end{definition}

One of the most common mechanism to reach $\epsilon$-DP is Laplace mechanism, as defined in the following.

\begin{definition}[Laplace Mechanism]
Given a function $f$ with sensitivity $\Delta f$, the Laplace mechanism ${M}_L$ adds i.i.d. noise from $\text{Lap}(\Delta f/\epsilon)$\footnote{The Laplace distribution $\text{Lap}(b)$ has probability density function $p(x) = \frac{1}{2b}e^{-\frac{|x|}{b}}$, where $b > 0$ is the scale parameter. This distribution has mean 0 and variance $2b^2$. In the context of differential privacy, the scale parameter $b = \frac{\Delta f}{\epsilon}$ determines the amount of noise added, with larger values providing stronger privacy guarantees but reducing accuracy.} to each output dimension:
\begin{equation}
    {M}_L(X) = f(X) + \text{Lap}\left(\frac{\Delta f}{\epsilon}\right).
\end{equation}
This mechanism satisfies $\epsilon$-DP.
\end{definition}

\subsection{CMC Model \cite{ching2002multivariate}}
Sequential and multivariate data streams usually exhibit both \emph{intra-} (self) and \emph{inter-} (cross) correlations. To capture this, we introduce a typical multivariate Markovian model--Coupling Markov Chain (CMC)\cite{ching2002multivariate}, an extension of time-homogeneous Markov chains to multiple correlated sequences, where a sequence generates data by one source over time.

We assume that there are $s$ categorical sequences and each has $m$ possible states in the state set $\{1, 2, \ldots, m\}$. Source \(i\in\{1,\dots,s\}\) generates a discrete-time sequence \(\boldsymbol{x}^{(i)}=(x^{(i)}_{1},x^{(i)}_{2},\dots,x^{(i)}_{n})\) over an alphabet \(\mathcal{X}=\{1,2,\ldots,m\}\). The system snapshot at slot \(t\) is the vector \(\boldsymbol{x}_{t}=(x^{(1)}_{t},\dots,x^{(s)}_{t})\), as shown in Fig.~\ref{fig:cmc_model}.

% Let $\boldsymbol{x}_n^{(k)}$ be the state vector of the $k$-th sequence at time $n$. If the $k$-th sequence is in State $j$ at time $n$ then $ \boldsymbol{x}_n^{(k)} = \mathbf{e}_j = (0, \ldots, 0, \underbrace{1}_{j\text{-th entry}}, 0\ldots, 0)^T$, as shown in Fig.~\ref{fig:cmc_model}.

\begin{figure}[t]
    \centering
    \includegraphics[width=0.48\textwidth]{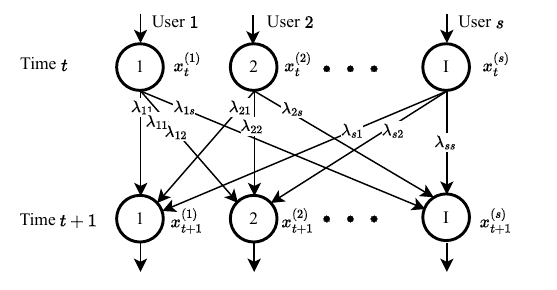}
    \caption{An illustration of CMC model.}
    \vspace{-15pt}
    \label{fig:cmc_model}
\end{figure}

Let $\{\boldsymbol{\pi}^{(k)}_n \}_{n=1}^{T}$ denote the state probability distribution vector of the $k$-th sequence at time $n$, where each sequence has $m$ discrete states. The coupling between sequences is defined as:
\begin{equation}
     \boldsymbol{\pi}^{(j)}_{n+1} = \sum_{k=1}^{s} \lambda_{jk} \boldsymbol{P}^{(jk)}  \boldsymbol{\pi}^{(k)}_n, \quad \forall j \in {1, \dots, s},
\end{equation}
where $\boldsymbol{P}^{(jk)} \in \mathbb{R}^{m \times m}$ is the transition matrix from the $k$-th sequence to the $j$-th sequence, and $\lambda_{jk} \ge 0$ with $\sum_{k=1}^s \lambda_{jk} = 1$.

Stacking all $s$ distributions into $\boldsymbol{\pi}_n = \left( \boldsymbol{\pi}^{(1)}_n, \dots,  \boldsymbol{\pi}^{(s)}_n\right)^{T}$, the system evolves as:
\begin{equation}
    \boldsymbol{\pi}_{n+1} = \boldsymbol{Q} \boldsymbol{\pi}_n,
\end{equation}
where $\boldsymbol{Q}$ is a block matrix composed of $\lambda_{jk} \boldsymbol{P}^{(jk)}$.

There are some key properties of CMC \cite{ching2002multivariate}:
\begin{itemize}
    \item \emph{Spectral Stability}: If $\lambda_{jk} > 0$ for all $j$ and $k$, the matrix $\boldsymbol{Q}$ has a dominant eigenvalue equal to $1$, with all other eigenvalues bounded by $1$ in modulus. This guarantees long-term stability.
    \item \emph{Steady-State Convergence}: If every $\boldsymbol{P}^{(jk)}$ is irreducible, there exists a unique stationary vector
    $\boldsymbol{\pi}^{\star} = [\boldsymbol{\pi}^{(1)}, \dots, \boldsymbol{\pi}^{(s)}]^{T}$
    % $\mathbf{x}^{\star}=[(\mathbf{x}^{(1)})^{\!\top},\dots,(\mathbf{x}^{(s)})^{\!\top}]^{\!\top}$
    with $\boldsymbol{\pi}^{\star}=\boldsymbol{Q}\,\boldsymbol{\pi}^{\star}$ and $\sum_{i=1}^m [\boldsymbol{\pi}^{(j)}]_i = 1$  for all $j$.  Moreover, $\lim_{n \to \infty} \boldsymbol{\pi}_n = \boldsymbol{\pi}^{\star}$.
    % Under irreducibility of $\boldsymbol{P}^{(jk)}$, there exists a unique steady-state vector $\boldsymbol{\pi} = [{\pi}^{(1)}, \dots, {\pi}^{(s)}]^{T}$ satisfying $\boldsymbol{\pi} = \boldsymbol{Q}\boldsymbol{\pi}$ and $\sum_{i=1}^m [\boldsymbol{\pi}^{(j)}]_i = 1$ for all $j$. Further, $\lim_{n \to \infty} \boldsymbol{\pi}_n = \boldsymbol{\pi}$.
\end{itemize}
These guarantees imply that a CMC converges to a single stationary
distribution that faithfully preserves both intra- and inter-sequence
dependencies, providing a solid probabilistic backbone for analyzing privacy leakage over time.

% 在上述现有文本之后添加：
\textbf{Significance for Privacy Analysis.} The stationary convergence property is crucial for our privacy analysis as it: (1) ensures our leakage bounds remain valid over extended time horizons by stabilizing the system's probabilistic behavior, and (2) enables characterization of correlation decay with age through spectral stability, directly informing how privacy protection improves with data staleness. This convergence guarantee provides the theoretical foundation for the Aging mechanism in our proposed algorithm.

\subsection{Age-Dependent Differential Privacy (ADP)\cite{Zhang2023}}

Classical DP assumes data freshness is irrelevant to privacy. However, in time-evolving data, outdated information is often less sensitive. ADP introduces a time-aware generalization.
\begin{definition}[ADP \cite{Zhang2023}]
    A mechanism $\mathcal{M}$ satisfies $(\epsilon(t), t)$-ADP if, for all pairs of datasets $X_t, X_t'$ differing in one individual's state at time $t$, and for all events $S \subseteq \mathcal{Y}$,
    \begin{equation}
        \label{eq:adp}
        \Pr[\mathcal{M}(X_0) \in S \mid X_t] \le e^{\epsilon(t)} \cdot \Pr[\mathcal{M}(X_0) \in S \mid X_t'].
    \end{equation}
\end{definition}
Here, $\epsilon(t)$ is a function of data age. As $t \to \infty$, data becomes stale and $\epsilon(t) \to 0$, implying less privacy risk.

At time \(t\) we publish one noisy and aged snapshot \(M(X_{0})\); inequality \eqref{eq:adp} then caps how much this release can sway an adversary's belief about the \emph{current} record \(X_{t}\) by the factor \(e^{\epsilon(t)}\). Natural dynamics weaken the link between \(X_{0}\) and \(X_{t}\), and the added noise blurs what remains, so the budget \(\epsilon(t)\) decays with \(t\) toward zero—allowing yesterday's data to be useful while shielding today's state.

\begin{proposition}[Privacy Leakgae Level of $(\epsilon(t), t)$-ADP \cite{Zhang2023}]
If $\mathcal{M}$ satisfies $\epsilon_C$-DP, and the underlying stochastic process has a maximal total variation distance $\Delta(t)$ at lag $t$, then $\mathcal{M}$ also satisfies ADP with
\begin{equation}
    \epsilon(t) = \ln\left(1 + \Delta(t) \cdot (e^{\epsilon_C} - 1)\right).
\end{equation}
\end{proposition}

This characterization shows how data aging can be used alongside noise injection to reduce privacy risk without compromising data utility severely. It provides a quantitative way to trade off accuracy and protection in dynamic systems.

% Classical DP assumes data freshness is irrelevant to privacy. However, in time-evolving data, outdated information is often less sensitive. ADP introduces a time-aware generalization.

% \begin{definition}[ADP \cite{Zhang2023}]
% A mechanism $\mathcal{M}$ satisfies $(\epsilon(t), t)$-ADP if, for all pairs of datasets $X_t, X_t'$ differing in one individual's state at time $t$, and for all events $S \subseteq \mathcal{Y}$,
% \begin{equation}
%     \Pr[\mathcal{M}(X_0) \in S \mid X_t] \le e^{\epsilon(t)} \cdot \Pr[\mathcal{M}(X_0) \in S \mid X_t'].
% \end{equation}
% \end{definition}

% Here, $\epsilon(t)$ is a function of data age. As $t \to \infty$, data becomes stale and $\epsilon(t) \to 0$, implying less privacy risk.

% \begin{proposition}[Privacy Leakgae Level of $(\epsilon(t), t)$-ADP \cite{Zhang2023}]
% If $\mathcal{M}$ satisfies $\epsilon_C$-DP, and the underlying stochastic process has a maximal total variation distance $\Delta(t)$ at lag $t$, then $\mathcal{M}$ also satisfies ADP with
% \begin{equation}
%     \epsilon(t) = \ln\left(1 + \Delta(t) \cdot (e^{\epsilon_C} - 1)\right).
% \end{equation}
% \end{proposition}

% This characterization shows how aging (i.e., data freshness) can be used alongside noise injection to reduce privacy risk without compromising data utility severely. It provides a quantitative way to trade off accuracy and protection in dynamic systems.

\section{Correlated-Sequence Differential Privacy (CSDP)}
In this section, we begin by introducing the setup of sequence correlation. Building on this foundation, we will define the CSDP framework and subsequently present the mechanism FRAN to achieve CSDP.

\subsection{Correlation Model of Sequences}

% In many real-world applications, such as sensor networks, IoT systems, or multi-user systems, data points across different sources are often highly correlated in both time and space. These correlations can significantly impact the privacy guarantees provided by traditional privacy mechanisms.
Practical data streams—ranging from multi-sensor IoT feeds to user activity logs—are rarely independent. They exhibit correlations both in time (within each stream) and in space (across streams), and ignoring these links leads to over-optimistic privacy claims.

We consider a scenario where multiple data sources, such as sensors or users, produce time-indexed data sequences that are interdependent. 
Let \(s\) be the number of correlated sources.
Source \(i\in\{1,\dots,s\}\) generates a discrete-time sequence
\(\boldsymbol{x}^{(i)}=(x^{(i)}_{1},x^{(i)}_{2},\dots,x^{(i)}_{T})\) over an alphabet
\(\mathcal{X}\).
The system snapshot at slot \(t\) is the vector
\(\boldsymbol{x}_{t}=(x^{(1)}_{t},\dots,x^{(s)}_{t})\in\mathcal{X}^{s}\).
Collecting all \(T\) snapshots yields the spatio-temporal database
\begin{equation}
    X_{1:T}\;=\;\bigl\{\boldsymbol{x}_{t}\bigr\}_{t=1}^{T}\in\mathcal{X}^{s\times T}.
\end{equation}

We then form the correlation structure from the following two types of dependencies:
\begin{itemize}
\item \emph{Intra-sequence (temporal):}
      for any source \(i\) and lag \(k\ge 1\),
      \(x^{(i)}_{t}\) and \(x^{(i)}_{t+k}\) are statistically dependent.
\item \emph{Inter-sequence (spatial):}
      for two distinct sources \(i\neq j\) and lag \(k\ge 1\),
      \(x^{(i)}_{t}\) and \(x^{(j)}_{t+k}\) may also be dependent.
\end{itemize}

Because adversaries can leverage these dependencies to infer sensitive information, privacy mechanisms must explicitly account for them rather than assume independent records.

% Let $\mathcal{S}$ represent a spatiotemporal domain composed of $s$ correlated data sources (e.g., multiple sensors, users, or data-generating entities). Each source $i \in \{1, \ldots, s\}$ generates a time-indexed sequence $x^{(i)} = (x_1^{(i)}, x_2^{(i)}, \ldots, x_T^{(i)})$ where $x_t^{(i)} \in \mathcal{X}$, and $\mathcal{X}$ is a discrete state space representing the possible values at each time step.

% The full system state at time $t$ is represented as the vector $\boldsymbol{x}_t = (x_t^{(1)}, \ldots, x_t^{(s)})$, where each component corresponds to one of the $s$ sequences. This model captures both temporal dependencies within individual sequences and spatial dependencies across sequences at each time step.

% Formally, we define the spatiotemporal database as:
% \begin{equation}
%     \boldsymbol{X}_T = \left(x_t^{(i)}\right)_{1 \leq t \leq T, 1 \leq i \leq s} \in \mathcal{X}^{T \times s}.
% \end{equation}

% We assume that the data exhibits both aspects of correlations as:
% \begin{itemize}
%     \item \emph{Intra Correlation}: Data points within a sequence are related over time, i.e., $x_t^{(i)} \rightarrow x_{t+k}^{(i)}$ for $k\ge 1$.
%     \item \emph{Inter Correlation}: Data points across different sequences may also exhibit relationships, i.e., $x_t^{(i)} \leftrightarrow x_{t+k}^{(j)}$ for $i \neq j$ and $k\ge 1$.
% \end{itemize}

% These dependencies are critical to preserve in privacy-preserving mechanisms, as they can be exploited by adversaries to infer sensitive information about individual data points. 

\begin{algorithm}[t]
\caption{FRAN: Freshness-Regulated Adaptive Noise for CSDP}
\label{alg:2step_mechanism}
\textbf{Input:} 
Multivariate time-series data \( X_{1:T} \), privacy budget \(\epsilon_S\), Laplace noise level \(\epsilon_C\), AoI vector \(\boldsymbol{A}_t = (A^{(1)}_t, \dots, A^{(s)}_t)\), query function \(f\).

\begin{algorithmic}[1]
\Statex \hspace{-15pt} {\textit{\%\% \textbf{Phase 1. Data-Aging Phase:}}}
% \State \textbf{Phase 1: Data-Aging Phase:}
\State Slide every sequence back by the AoI vector:
\begin{equation}
    \tilde{\boldsymbol{x}}_{t} \;=\;
\bigl(x^{(1)}_{\,t-A^{(1)}_{t}},\dots,x^{(s)}_{\,t-A^{(s)}_{t}}\bigr).
\end{equation}
% \State Larger \(A^{(i)}_t\) values reduce the maximal total-variation distance \(\Delta(t)\), shrinking the required privacy budget.

\Statex {}
\Statex \hspace{-15pt} {\textit{\%\% \textbf{Phase 2. Noise-Injection Phase:}}}
% \State \textbf{Phase 2: Noise-Injection Phase:}
\State After aging, apply the Laplace mechanism to the aged data \(\tilde{\boldsymbol{x}}_{t}\) for any query function \(f: \mathcal{X}^s \to \mathbb{R}^d\):
\begin{equation} \label{eq:mechanism}
    M_S(X_{1:T}) = f(\tilde{\boldsymbol{x}}_{t}) + \eta, \quad \eta \sim \text{Lap} \left(\frac{\Delta_f}{\epsilon_C}\right),
\end{equation}
where \(\Delta_f\) is the sensitivity of the query \(f\) and \(\epsilon_C\) is the Laplace noise level.
\end{algorithmic}

\textbf{Output:} FRAN mechanism output \(M_S(X_{1:T})\).
\end{algorithm}
\vspace{-15pt}

\subsection{Definition of CSDP}

Classic differential privacy assumes independent records; here we extend the guarantee to spatio-temporal datasets in which records are both time-ordered and cross-correlated.

% CSDP extends the concept of traditional Differential Privacy (DP) to the setting of correlated sequential data. The goal of CSDP is to ensure that the privacy of individual data points is preserved while accounting for the correlations that exist between different data sources.
\begin{definition}[Neighbouring spatio-temporal databases]
    Fix a horizon \(T\) and \(s\) sources.  Let \(X_{1:T},\,X'_{1:T}\in\mathcal{X}^{s\times T}\) be two datasets.
    We say they are \emph{\(t\)-neighbour}, written \(X_{1:T}\overset{t}{\sim}X'_{1:T}\), if  $\boldsymbol{x}_{t}\in X_{1:T}$ and $\boldsymbol{x}'_{t}\in X'_{1:T}$ differ in \emph{exactly one} entry for the given time \(t\): there exist \(i_{0}\in\{1,\dots,s\}\) such that
% \[
% x^{(i)}_{u}=x'^{(i)}_{u}\quad\forall(i,u)\neq(i_{0},t),\qquad
% x^{(i_{0})}_{t}\neq x'^{(i_{0})}_{t}.
% \]
% Two sequence datasets $\boldsymbol{x}_t $ and $\boldsymbol{x}_t '$ are considered neighboring in time $t$, denoted as $\boldsymbol{x}_t \sim \boldsymbol{x}_t'$, if they differ in time $t$ of a single individual's sequence. Specifically, there exists an index $i_0$ such that for time $t$:
    \begin{equation}
        x_t^{(i)} = x_t^{\prime(i)} \quad \forall i \neq i_0, \quad \quad x_t^{(i_0)} \neq x_t^{\prime(i_0)}.
    \end{equation}
\end{definition}
% This definition captures the notion that a change in one individual's data across all time steps leads to a neighboring dataset.

Let \(\mathbf{\Lambda}\) be the family of joint distributions that capture the
temporal and spatial dependencies described in Section~III.  
All probabilities below are taken with respect to an arbitrary \(\pi\in\mathbf{\Lambda}\).

\begin{definition}[CSDP]
    \label{def:csdp}
    A randomised mechanism
    \(M:\mathcal{X}^{s\times T}\!\rightarrow\!\mathcal{Y}\)
    satisfies \emph{\((\epsilon_{S},\mathbf{\Lambda})\)-CSDP} if, for every
    \(\pi\in\mathbf{\Lambda}\), every time index \(t\), every pair of
    \(t\)-neighbouring datasets \(X_{1:T}\overset{t}{\sim}X'_{1:T}\), and every
    measurable output set \(S\subseteq\mathcal{Y}\),
\begin{equation}
    \Pr[M(X_{1:T})\in S | \boldsymbol{x}_t ]\;\le\;
    e^{\epsilon_{S}}\,
    \Pr[M(X'_{1:T})\in S | \boldsymbol{x}'_t].
    \label{eq:csdp}
\end{equation}
\end{definition}

The leakage parameter \(\epsilon_{S}\) plays the usual role: smaller values
yield stronger privacy but may degrade utility. Inequality~\eqref{eq:csdp} ensures that, even under any dependence model in \(\mathbf{\Lambda}\), observing the mechanism's output barely changes an adversary's belief about the single fresh record that differs between the two datasets.

CSDP naturally generalizes existing privacy frameworks.
\begin{lemma}[Generalization of ADP and DP]\label{lemma:csdp}
   \begin{itemize}
       \item With only independent distributions, $(\epsilon_S, \mathbf{\Lambda})$-CSDP reduces to $\epsilon_S$-DP.
       \item With only temporal correlations, $(\epsilon_S, \mathbf{\Lambda})$-CSDP reduces to $(\epsilon(t), t)$-ADP where $\epsilon_S = \max_t \epsilon(t)$.
   \end{itemize}
\end{lemma}
\begin{proof}[Proof Sketch]
   For independent distributions, conditioning becomes irrelevant, yielding standard DP. For temporal-only correlations, CSDP independently guarantees each sequence's privacy, equivalent to ADP. See Appendix~\ref{appendix:lemma1} for details.
\end{proof}
% Privacy leakage level ${\epsilon_S}$ controls the privacy guarantee: smaller values of ${\epsilon_S}$ provide stronger privacy protection but may reduce data utility.

\subsection{FRAN Mechanism Design}
\label{subsec:2step}

Protecting correlated streams requires more than simply adding Laplace noise: the strength of the privacy guarantee depends jointly on \emph{how fresh} the published data are and \emph{how much} randomisation is applied. We therefore propose \emph{FRAN} (Freshness-Regulated Adaptive Noise), which separates the mechanism into two explicit phases.

\subsubsection{FRAN Mechanism Design}

As shown in Algorithm~\ref{alg:2step_mechanism} and Fig.~\ref{fig:2step_mechanism}, FRAN operates in two phases:

\textbf{(Phase 1. Data-aging.)}  
Recent measurements carry higher privacy risk due to strong correlations with the current state. To reduce this, each sequence is shifted by an \emph{Age-of-Information (AoI)} vector \( \boldsymbol{A}_t = (A^{(1)}_t, \dots, A^{(s)}_t) \), weakening the correlation between the published data and the current state. Larger \(A^{(i)}_t\) values reduce the maximal total-variation distance \( \Delta(t) \), thus lowering the required privacy budget. However, excessive aging decreases data timeliness, requiring a balance between AoI and noise level.

\textbf{(Phase 2. Noise-injection.)}  
After aging, we apply the Laplace mechanism to the query function \( f \) with correlation-aware sensitivity, computed as the worst-case change in \( f \) when two inputs differ by one allowed entry (\(i_0, t\)) under the dependency model \( \mathbf{\Lambda} \). The noise scale \( \epsilon_C \) is the Laplace noise level, with higher values reducing noise (less privacy protection but better utility).

\begin{figure}[!t]
    \centering
    \includegraphics[width=0.5\textwidth]{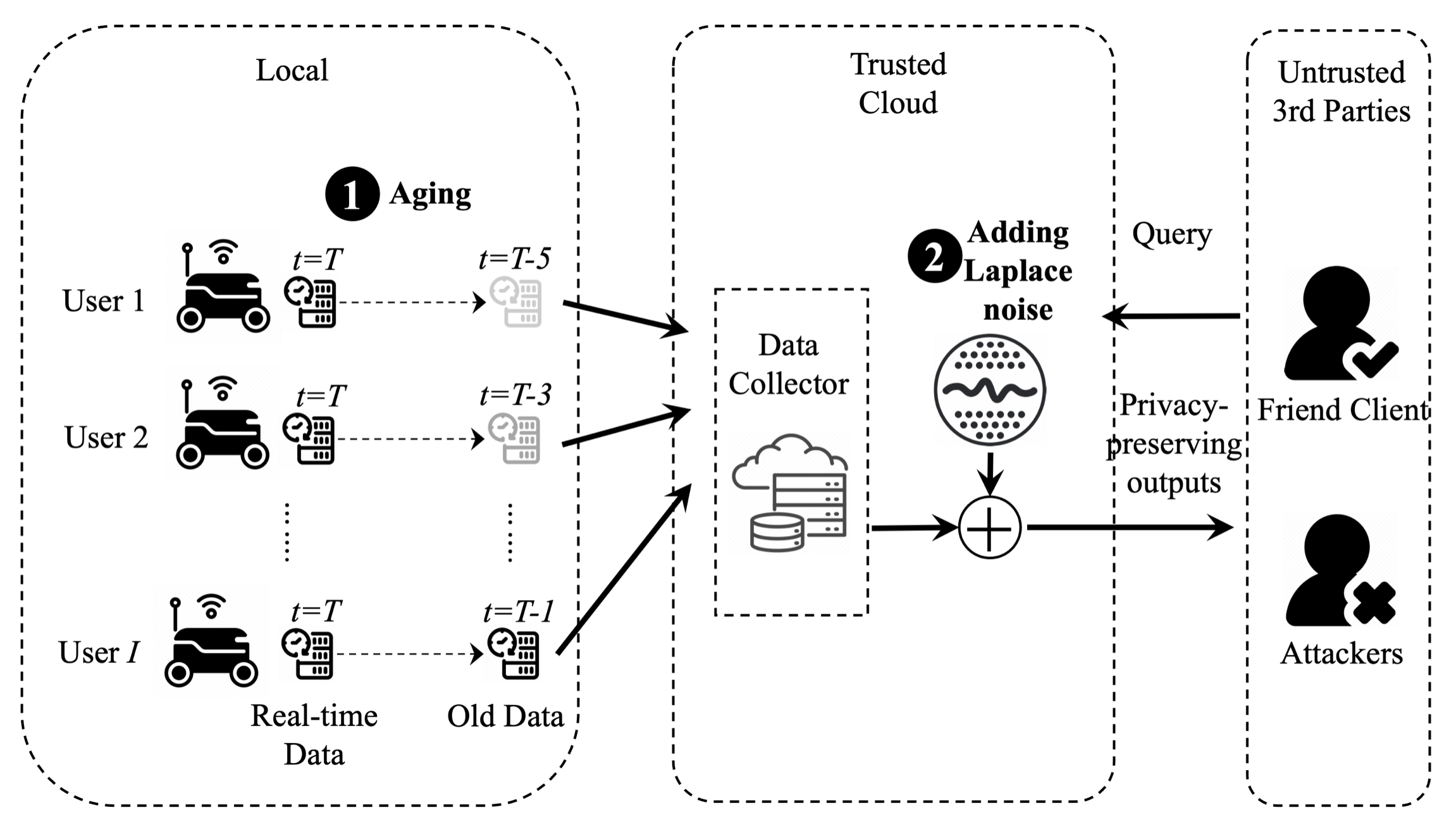}
    \caption{The illustration of FRAN mechanism for CSDP.}
    \vspace{-13pt}  % 缩减图像与正文之间的距离
    \label{fig:2step_mechanism}
\end{figure}

\subsubsection{Complexity Analysis}
The FRAN mechanism achieves linear-time complexity $O(sd)$, where $s$ is the number of sequences and $d$ is the output dimension. Specifically, Phase 1 (data aging) requires $O(s)$ operations to shift each sequence by its corresponding AoI value. Phase 2 (noise injection) requires $O(d)$ operations to add Laplace noise to the $d$-dimensional query output. Since both phases process each sequence independently and avoid complex correlation computations at runtime, FRAN scales efficiently to large multivariate time series datasets with thousands of sequences.

\subsubsection{Interplay Between Aging and Noise}

Equation~\eqref{eq:mechanism} shows the two controls side-by-side:
larger AoI increases staleness but permits a smaller \(\epsilon_C\), whereas
reducing AoI improves freshness but demands either stronger noise or a larger
privacy parameter \(\epsilon_{S}\).  The mechanism designer can therefore
tune \((\boldsymbol{A}_{t},\epsilon_C)\) to meet an application-specific
freshness-utility target under a fixed privacy budget, or conversely hold
utility constant and tighten privacy by modestly aging the input.

% The FRAN mechanism reveals a fundamental trade-off between data aging and noise injection:
% \begin{itemize}
% \item \textbf{Aging Control}: Larger AoI increases staleness but permits a smaller \(\epsilon_C\)
% \item \textbf{Noise Control}: Reducing AoI improves freshness but demands either stronger noise or a larger privacy parameter \(\epsilon_{S}\)
% \end{itemize}

% The mechanism designer can therefore tune the parameter pair (\(\boldsymbol{A}_{t},\epsilon_C\)) to meet application-specific requirements under a fixed privacy budget, or conversely hold utility constant and tighten privacy by modestly aging the input.

\subsubsection{Behavior in Extreme Cases}
The FRAN mechanism exhibits interesting properties at the extremes:
\begin{itemize}
\item \textbf{Fresh-Data Regime} (\(\boldsymbol{A}_{t}=\boldsymbol{0})\): With no aging, FRAN falls to the Laplace mechanism. Privacy may deteriorate quickly when streams are strongly coupled.

\item \textbf{Stale-Data Regime} (Large \(\boldsymbol{A}_{t}\) where \(\Delta(t)\approx 0\)): The release becomes almost independent of the present state, requiring minimal noise—but the information becomes correspondingly stale.
\end{itemize}

In summary, FRAN adapts classical DP to the correlated, time-evolving setting by \emph{first} throttling the freshness of the data and \emph{then} adding just enough noise to cover any residual leakage, yielding a flexible privacy-freshness-accuracy trade-off tailored to correlated sequential data.

\section{Privacy-Leakage Bounds under CSDP}
\label{sec:leakage}
This section quantifies how much information a CSDP-compliant release can
reveal.  We present two complementary results:
\begin{itemize}
    \item \textbf{Loose bound (Section V.A).}  A closed-form expression that is esay to evaluate and therefore attractive for quick parameter tuning.  It yields a safety margin—the true leakage never exceeds this bound—at the price of some conservatism.
    \item \textbf{Tight bound (Section V.B).}  A sharper, optimisation-based bound that tracks the dependency structure more closely.  It offers tighter guarantees but requires solving a small convex programme or Monte-Carlo estimator.
\end{itemize}
We show when each bound is preferable and illustrate the accuracy-complexity
trade-off with numerical examples.

\subsection{Loose Privacy--Leakage Bound}
\label{subsec:loose}
This subsection develops a quick, conservative upper bound on the CSDP
leakage parameter~$\epsilon_{S}$.  
The key idea is to decompose leakage into four intuitive factors:
(i) how many users are directly correlated,  
(ii) how sensitive the query is to those users,  
(iii) how much correlation remains after aging, and  
(iv) the scale of the Laplace noise.
%--------------------------------------------------
\begin{definition}[Correlation degree $k$]
\label{def:k}
Under the dependency model~$\mathbf{\Lambda}$, every user record is directly
correlated with at most $k-1$ other users.  The integer $k$ is called the
\emph{correlation degree}.
\end{definition}
The correlation degree captures the width of the dependency graph.
When $k=1$ the users are independent; larger values indicate tighter coupling
and therefore greater potential privacy risk.
%--------------------------------------------------
\begin{definition}[$k$-sensitivity $d(k)$]
\label{def:dk}
Let $f:\mathcal{X}^{s}\!\rightarrow\!\mathbb{R}^{d}$ be a query.
For $i\in\{1,\dots,k\}$ define
\begin{equation}
    d(k) \triangleq \frac{s_k{(f)}}{s_1{(f)}},
\end{equation}
where $s_{i}(f)=
\max_{X,X'\colon dist(X,X')=i}
\lVert f(X)-f(X')\rVert_{1}$ counts the query sensitivity on the number of user records that differ, and $dist(\cdot,\cdot)$ counts this number.
\end{definition}

The function $d(k)$ quantifies how strongly the query output can change when
several correlated users modify their data.  For example,
\(d(k)=k\) for the \textsc{sum} or \textsc{mean} opeations, and
\(d(k)=1\) for the \textsc{max} or \textsc{min} opeations.
%--------------------------------------------------
\begin{definition}[Aged correlation distance $\Delta_{k}$]
\label{def:tvk}
Let $\boldsymbol{x}^{\mathcal{K}}_{t}$ and
$\boldsymbol{x}'^{\mathcal{K}}_{t}$ be size-$k$ snapshots that differ in
the positions allowed by Definition~\ref{def:k}.
After aging by the AoI vector~$\boldsymbol{A}_{t}$, their maximal
total-variation distance is
\begin{equation}
\Delta_{k}(\mathbf{\Lambda},\boldsymbol{A}_{t})=
\max_{\boldsymbol{x}^{\mathcal{K}}_{t},\,
      \boldsymbol{x}'^{\mathcal{K}}_{t}} \delta
      \left( 
    \Pr[\boldsymbol{x}^{\mathcal{K}}_{t-A^{\mathcal{K}}_{t}}\mid\boldsymbol{x}^{\mathcal{K}}_{t}]
    - \Pr[\boldsymbol{x}^{\mathcal{K}}_{t-A^{\mathcal{K}}_{t}}\mid\boldsymbol{x}'^{\mathcal{K}}_{t}]
    \right),
\end{equation}
where $\delta(P, Q)=\sup_{X \in \mathcal{F}} |P(X) - Q(X)|$ is the total variation distance.
\end{definition}
The distance $\Delta_{k}$ measures how much information about the present
survives in the aged snapshot.  A larger AoI vector makes this distance
smaller because past states are less predictive of the current ones.
%--------------------------------------------------
\begin{theorem}[Loose leakage bound]
\label{thm:loose}
Consider mechanism FRAN \ref{alg:2step_mechanism} with AoI
vector~$\boldsymbol{A}_{t}$ and Laplace noise level~$\epsilon_{C}$.
FRAN satisfies $(\epsilon_{S},\mathbf{\Lambda})$-CSDP with
\begin{equation}
\epsilon_{S}(\boldsymbol{A}_{t},\epsilon_{C})=
d(k)\,
\Delta_{k}(\mathbf{\Lambda},\boldsymbol{A}_{t})\,
\epsilon_{C}.
\label{eq:loose}
\end{equation}
\end{theorem}
\begin{proof}[Proof sketch]
Couple any pair of $t$-neighbouring databases so that at most $k$ correlated
users differ, apply the standard Laplace-DP argument with the
$k$-sensitivity~$d(k)$, and upper-bound the residual dependence by
$\Delta_{k}$.  Multiplying the three factors yields~\eqref{eq:loose}. For the details of the proof, see Appendix~\ref{appendix:thm1}.
\end{proof}

Equation~\eqref{eq:loose} highlights the design trade-offs.
Increasing AoI shrinks $\Delta_{k}$, thereby reducing leakage but at the cost
of staler information.  Decreasing the noise level~$\epsilon_{C}$ improves
utility yet enlarges $\epsilon_{S}$.  Finally, queries with large
$k$-sensitivity $d(k)$ are intrinsically harder to protect.  Although the
bound may be pessimistic when dependencies are weak, its closed-form nature
makes it valuable for preliminary tuning and for applications requiring
worst-case guarantees.

\subsection{Tight Privacy--Leakage Bound}
\label{subsec:tight}
The loose bound in Section~\ref{subsec:loose} is quick to evaluate but can
grossly overestimate leakage when correlations are weak.
Here we tighten the estimate by working directly with conditional
distributions, at the cost of more computation and a need for finer knowledge
of the dependency model~$\mathbf{\Lambda}$.
%--------------------------------------------------
\begin{definition}[Bounded aged correlation $\bar{\Delta}$]
\label{def:bdtv}
Fix an AoI vector $\boldsymbol{A}_{t}$.
For any user $i$, let $\boldsymbol{x}_{-i}$ collect all other users' data.
The \emph{maximal bounded total-variation distance} is
\begin{align} \label{eq:bdtv}
    \bar{\Delta}(\mathbf{\Lambda}, \boldsymbol{A}_t) &= \max_{(\boldsymbol{x}_{-i}, x_i), (\boldsymbol{x}_{-i}, x_i')} \delta \Big( \bar{g}(\boldsymbol{x}_{-i}, x_i) \Pr[z_i \mid x_i], \nonumber\\ 
    &\quad \underline{g}(\boldsymbol{x}_{-i}, x_i') \Pr[z_i \mid x_i'] \Big),
\end{align}

where:
\begin{equation}
    \bar{g}(\boldsymbol{x}_{-i}, x_i) = \max_{\boldsymbol{z}_{-i}, z_i} \frac{\Pr[x_i \mid \boldsymbol{x}_{-i}, \boldsymbol{z}_{-i}] \Pr[\boldsymbol{x}_{-i} \mid x_i, z_i]}{\Pr[x_i \mid \boldsymbol{x}_{-i}] \Pr[\boldsymbol{x}_{-i} \mid x_i]},
\end{equation}
and
\begin{equation}
    \underline{g}(\boldsymbol{x}_{-i}, x_i) = \min_{\boldsymbol{z}_{-i}, z_i} \frac{\Pr[x_i \mid \boldsymbol{x}_{-i}, \boldsymbol{z}_{-i}] \Pr[\boldsymbol{x}_{-i} \mid x_i, z_i]}{\Pr[x_i \mid \boldsymbol{x}_{-i}] \Pr[\boldsymbol{x}_{-i} \mid x_i]},
\end{equation}
% where the simplified-form aged variables $z_{i}$ and $\boldsymbol{z}_{-i}$ correspond to the simplified-form variables
% $x_{i}$ and $\boldsymbol{x}_{-i}$ shifted back by $\boldsymbol{A}_{t}$, for better readability.
where aged variables $z_{i}$ and $\boldsymbol{z}_{-i}$ correspond to $x_{i}$ and $\boldsymbol{x}_{-i}$ shifted back by the AoI vector $\boldsymbol{A}_{t}$. All of them are in the simplified form  for better readability.
\end{definition}
Equation~\eqref{eq:bdtv} measures the worst mismatch between \emph{conditional}
distributions that differ at only one fresh record while honouring all
dependencies and the chosen AoI policy.
%--------------------------------------------------
\begin{theorem}[Tight leakage bound]
\label{thm:tight}
FRAN satisfies $(\epsilon^{\text{tight}}_{S},\mathbf{\Lambda})$-CSDP with
\begin{equation}
\epsilon_{S}^{\text{tight}}(\boldsymbol{A}_{t},\epsilon_{C})=
\bar{\Delta}(\mathbf{\Lambda},\boldsymbol{A}_{t})\,
\epsilon_C.
\label{eq:tight}
\end{equation}
\end{theorem}
\begin{proof}[Proof sketch]
Start from Definition~\ref{def:csdp}, expand each likelihood ratio, bound it
by the distance in Definition~\ref{def:bdtv}, and multiply by the Laplace
noise factor $\epsilon_C$ to obtain~\eqref{eq:tight}. To save space, we put the detailed proof of Theorem~\ref{thm:tight} in Appendix~\ref{appendix:thm2}.
\end{proof}

When correlations decay rapidly with AoI, the tight bound in
\eqref{eq:tight} can be orders of magnitude smaller than the loose bound,
yielding a less conservative privacy budget.
Computing $\bar{\Delta}$, however, involves either a convex programme
(in Markov settings) or Monte-Carlo sampling, so runtime grows with the
alphabet and network size.
A practical workflow is to design parameters with the loose bound for speed,
then certify the final choice with the tight bound to avoid unnecessary
noise.

\subsection{Comparative Discussion and Practical Implications}
\label{subsec:compare}

The two privacy–leakage bounds derived above are complementary rather than
competitive.  Both satisfy CSDP, yet each serves a different purpose because
they rely on different levels of information.

\begin{itemize}
\item \textbf{Loose bound (Theorem~\ref{thm:loose}).}  
      It depends only on four coarse-grained quantities:
      the correlation degree \(k\),
      the \(k\)-sensitivity \(d(k)\) of the query,
      the aged total-variation distance \(\Delta_{k}\),
      and the Laplace scale \(\epsilon_{C}\).
      None of these requires full knowledge of the joint distribution
      \(\mathbf{\Lambda}\), so the bound can be evaluated in closed form and
      in negligible time—ideal for rapid parameter sweeps or on-device
      adjustment.

\item \textbf{Tight bound (Theorem~\ref{thm:tight}).}  
      It replaces \(d(k)\,\Delta_{k}\) with the conditional metric
      \(\bar{\Delta}(\mathbf{\Lambda},\boldsymbol{A}_{t})\),
      computed under the complete dependency model and the chosen
      AoI vector.
      This yields a much sharper estimate—often one to two orders of
      magnitude smaller when correlations decay quickly with AoI.
      The drawback is computational: evaluating \(\bar{\Delta}\) entails a
      convex programme (for Markov models) or Monte-Carlo sampling, with
      runtime that grows with the alphabet size and the number of sources.
\end{itemize}

These observations suggest a two-phase engineering practice.  
First, employ the loose bound as a \emph{screen} to discard parameter sets
that clearly violate the privacy budget.  
Then apply the tight bound to the survivors, injecting noise only where it
is truly needed.  
This strategy achieves the desired budget without paying the high
computational cost of the tight analysis across the entire design space.

\section{Privacy-Utility Trade-off under CMC}

In this section, we present a comprehensive analysis of the privacy-utility trade-off in CSDP, with specific focus on CMC models. Our investigation begins with theoretical foundations connecting privacy leakage to utility degradation within the CMC framework. We then analyze how noise introduction affects utility, measured through Mean Squared Error (MSE). The section concludes with investigating optimization strategies designed to achieve an optimal balance between privacy preservation and data utility.

% This section presents an analysis of the privacy-utility trade-off in CSDP, focusing on CMC models. We connect privacy leakage to utility degradation within the CMC framework, analyze how noise affects utility through Mean Squared Error (MSE), and investigate optimization strategies to balance privacy preservation and data utility.

\subsection{Privacy Leakage Level under the CMC Model}
\label{subsec:cmc-leakage}
Under the coupling Markov chain (CMC) model described in Section~III.B, the privacy leakage and data utility of the FRAN mechanism are inherently linked to three key factors: (i)~the coupling strength between data sequences, (ii)~the spectral properties of the underlying Markov chain, and (iii)~the magnitude of the noise injected to protect privacy.

The primary challenge is identifying an optimal balance between noise level (which preserves privacy) and the utility of released data for downstream tasks, such as prediction or classification. To formalize this, we first analyze the combined effects of correlation and aging on the total-variation distance parameter.

% The challenge is finding an optimal balance between noise level and data utility for downstream tasks. To formalize this, we analyze the combined effects of correlation and aging on the total-variation distance parameter.

\begin{definition}[CMC Leakage Coefficient]
\label{def:cmc-coeff}
Given an AoI vector $\boldsymbol{A}_{t}$ at time $t$ and a coupling strength parameter $\lambda$, the \emph{CMC leakage coefficient} is defined as
\begin{equation}
    \Phi(\lambda,\boldsymbol{A}_{t}) = \Delta_{k}\bigl(\mathbf{\Lambda}(\lambda),\boldsymbol{A}_{t}\bigr),
\end{equation}
where $\mathbf{\Lambda}(\lambda)$ denotes the family of joint distributions induced by all CMCs whose coupling strengths are bounded by $\lambda$, and $\Delta_{k}$ is the aged total-variation distance from Definition~\ref{def:tvk}.
\end{definition}

\subsection{Privacy Leakage under CMC}

To specialize the loose leakage bound in Theorem~\ref{thm:loose} for Coupling Markov Chains (CMC), we define a model-aware leakage coefficient \( \Phi(\lambda, \boldsymbol{A}_t) \), which captures the residual correlation between aged data and the current state under coupling strength \( \lambda \).

\begin{proposition}[Privacy Leakage under CMC]
\label{prop:cmc-leakage}
FRAN in Algorithm~\ref{alg:2step_mechanism} satisfies \( (\epsilon_{S}, \mathbf{\Lambda}) \)-CSDP under CMC with
\begin{equation}
\epsilon_{S}^{\text{CMC}}(\boldsymbol{A}_t, \epsilon_C) = d(k)\, \Phi(\lambda, \boldsymbol{A}_t)\, \epsilon_C.
\label{eq:cmc-leakage}
\end{equation}
\end{proposition}

\begin{proof}[Proof sketch]
Directly follows from Theorem~\ref{thm:loose} by setting \( \Phi(\lambda, \boldsymbol{A}_t) := \Delta_k(\Lambda(\lambda), \boldsymbol{A}_t) \). See Appendix~\ref{appendix:prop2}.
\end{proof}

Compared to Equation~\eqref{eq:loose}, this formulation makes explicit the role of CMC parameters. It highlights that stronger coupling or fresher data increases leakage, while larger AoI or weaker dependencies reduce it—enabling more informed privacy-utility trade-offs.

\subsection{Utility Function and Impact of Privacy Noise}
\label{subsec:utility-noise}

To measure the utility loss caused by the privacy-preserving mechanism, we use the mean squared error (MSE) between the true data query outcome and the corresponding noisy result produced by the mechanism. Specifically, given a query function~$f(\cdot)$ applied to the true data~$\boldsymbol{x}_t$, the MSE at time~$t$ with noise scale~$\epsilon_{C}$ is defined as:
\begin{equation}
l_{\mathrm{MSE}}(\boldsymbol{A}_t,\epsilon_{C}) = \mathbb{E}\left[\,\|M_{S}(X_{1:t})-f(\boldsymbol{x}_t)\|^{2}\right],
\label{eq:mse}
\end{equation}
where the expectation is taken over the randomness introduced by mechanism FRAN.

% From equation~\eqref{eq:mse}, it is clear that decreasing~$\epsilon_{C}$, which corresponds to injecting more noise for stronger privacy guarantees, directly results in a higher MSE and thus lower utility. In other words, increasing privacy inevitably involves a utility trade-off, as noisier outputs degrade the accuracy of downstream tasks.

% In the context of CSDP, this utility degradation is further amplified by the coupling among data sequences. Stronger correlations mean that even moderate noise can severely disrupt the underlying relational structure between sequences, thereby inflating the MSE disproportionately. Conversely, weaker coupling between sequences tends to mitigate this impact, resulting in smaller utility losses for the same noise level.

% Therefore, carefully understanding the interplay between coupling strength and privacy noise is critical. Mechanism designers must balance these factors, adjusting the noise levels in accordance with the correlation strength, to ensure that the data utility remains acceptable while providing adequate privacy protection.

From equation~\eqref{eq:mse}, decreasing~$\epsilon_{C}$ (injecting more noise) directly results in higher MSE and lower utility. In CSDP, this utility degradation is amplified by coupling among data sequences. Stronger correlations mean even moderate noise can severely disrupt the underlying relational structure, inflating the MSE disproportionately. Weaker coupling tends to mitigate this impact, resulting in smaller utility losses for the same noise level.

Therefore, understanding the interplay between coupling strength and privacy noise is critical. Mechanism designers must balance these factors, adjusting noise levels according to correlation strength, to maintain acceptable data utility while providing adequate privacy protection.

\subsection{Optimization Problem Formulation}
Our optimization strategy aims to minimize privacy leakage while maintaining Mean Squared Error (MSE) within acceptable bounds for practical applications. We formulate this as Problem P1:
\begin{align}
\text{P1:} \quad \min_{\boldsymbol{A}_t,\epsilon_{C}} \quad & \epsilon_S(\boldsymbol{A}_t, \epsilon_C) \\
\text{s.t.} \quad & l_{\text{MSE}}(\boldsymbol{A}_t, \epsilon_C) \leq \bar{l}
\end{align}
where $l_{\text{MSE}}(\boldsymbol{A}_t, \epsilon_C)$ represents the expected MSE, and $\bar{l}$ denotes the target utility threshold. 

The objective is to minimize privacy leakage while ensuring data utility remains within specified bounds. By solving this optimization problem, we can determine the optimal noise parameters and sequence segmentation strategies that effectively balance privacy and utility, creating an adaptive mechanism that respects the data's structural characteristics while maintaining privacy guarantees.

In the numerical study section, we present a comprehensive optimization approach that balances privacy and utility within the CSDP framework.

\begin{figure*}[t]
    \centering
    \subfloat[]{\includegraphics[width=0.32\textwidth]{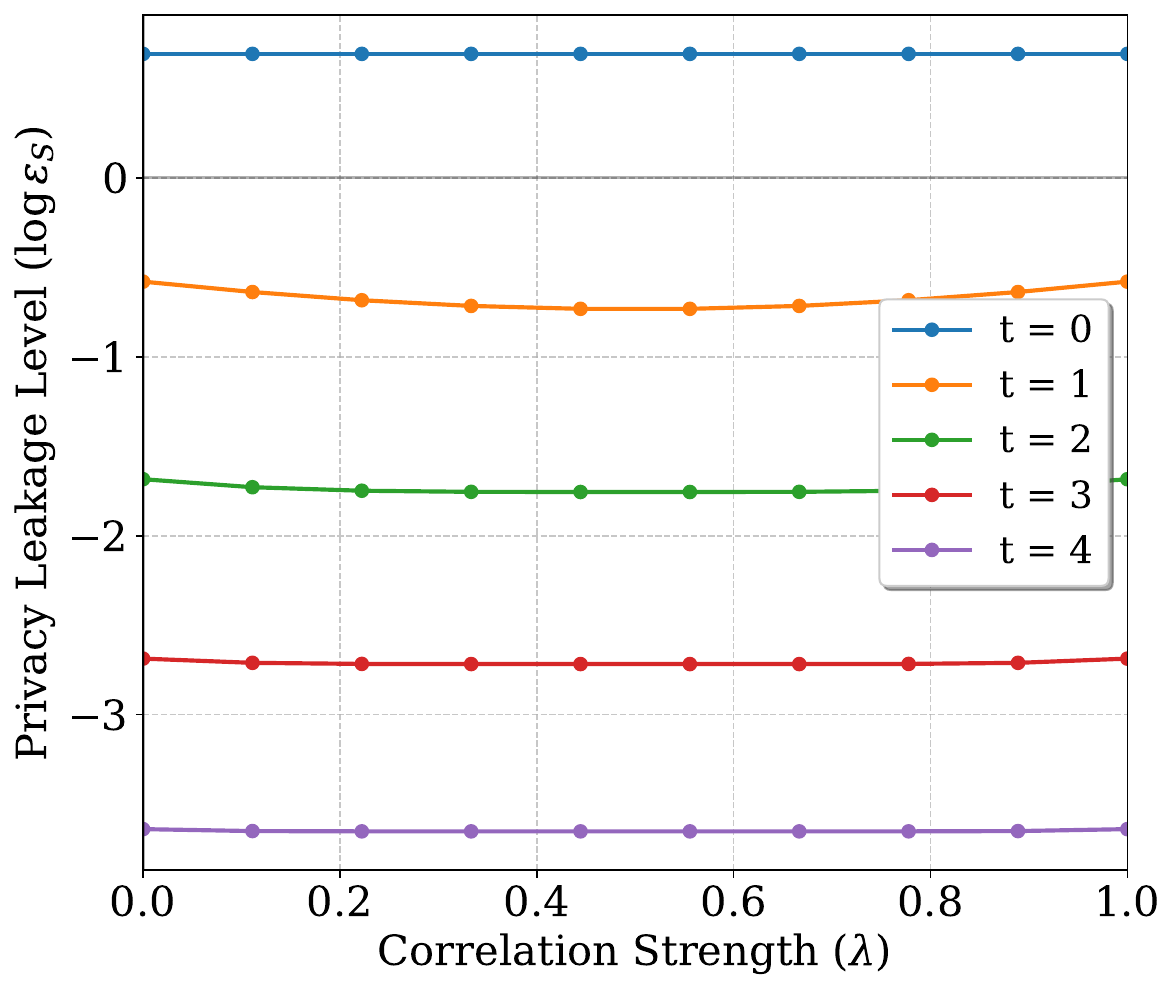}}
    % \vfill
    \subfloat[]{\includegraphics[width=0.32\textwidth]{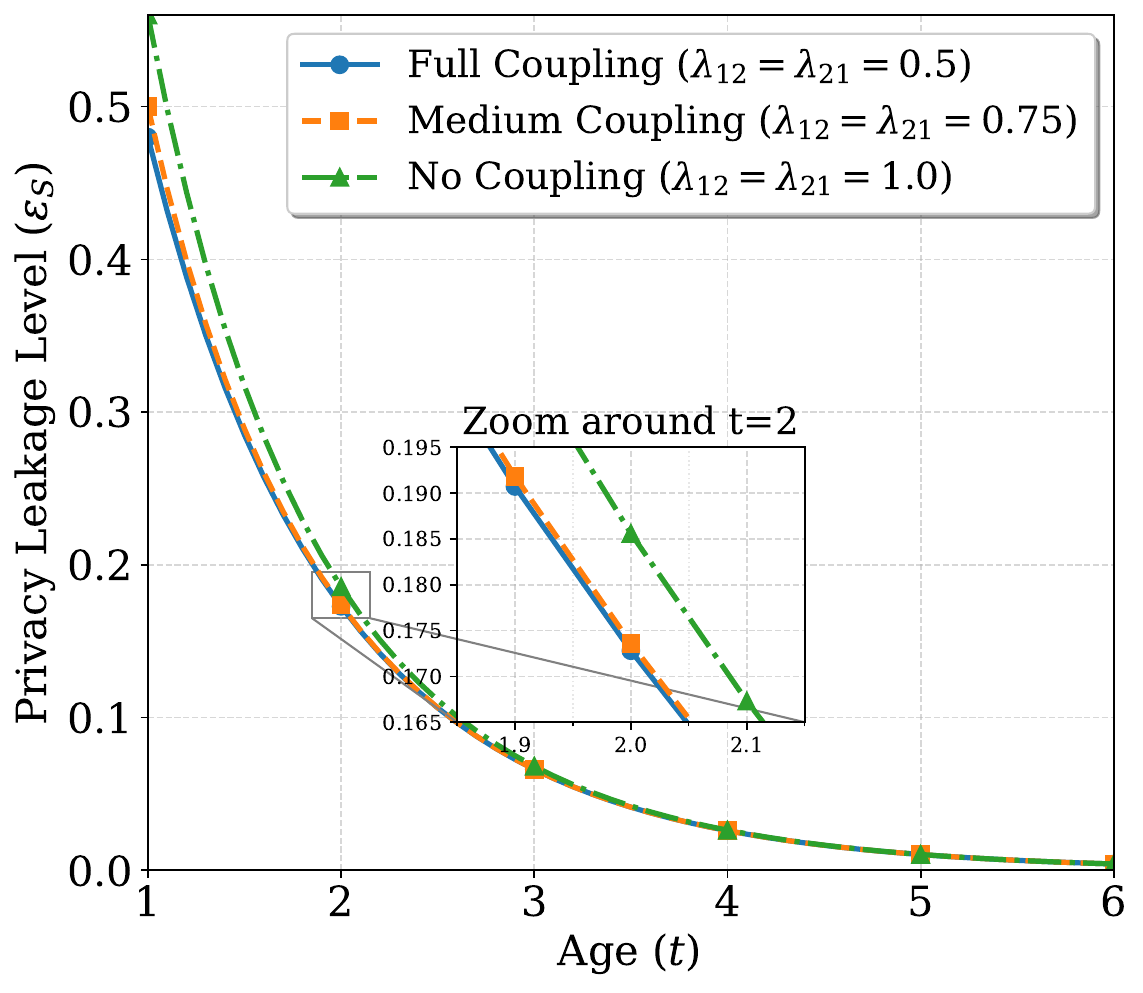}}
    % \hfill
    \subfloat[]{\includegraphics[width=0.32\textwidth]{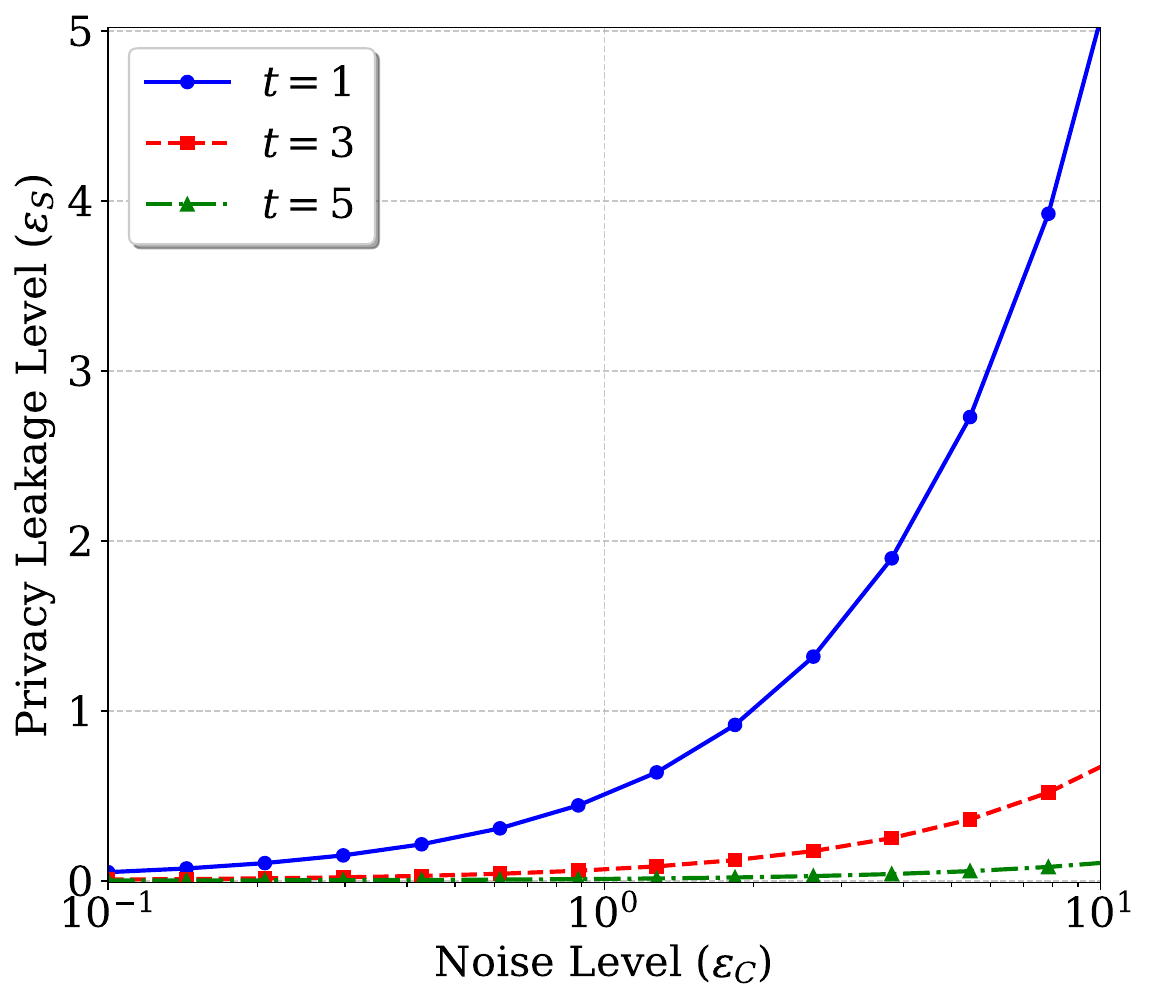}}
    \caption{Privacy leakage level v.s. different parameters. We set $p = q = 0.3$ and $\lambda = \lambda_{11} = \lambda_{22}$ across all experiments.}
    \vspace{-5pt}  % 缩减图像与正文之间的距离
    \label{fig:combined_analysis}
\end{figure*}

\section{Numerical Study}
\label{sec:numerical}

% In this section, we validate our theoretical findings through numerical experiments on a simplified scenario: two users with binary states, representing the core case of our CSDP framework.

In this section, we empirically validate our theoretical results through numerical experiments on a fundamental case study: a two-user system with binary states. This simplified scenario represents the essential elements of our CSDP framework while allowing for clear illustration of the key principles.

\subsection{Experimental Setup}
\label{subsec:setup}

We consider a system with two users ($s = 2$) each generating a binary sequence ($m = 2$) with states $\{0, 1\}$. We model the correlation using our CMC framework with transition matrices $P^{(jk)}$ and coupling parameters $\lambda_{jk}$ and set:
\begin{equation}
P^{(11)} = P^{(12)} = P^{(21)} =  P^{(22)} = 
\begin{bmatrix}
0.7 & 0.3 \\
0.3 & 0.7
\end{bmatrix}.
\end{equation}

To examine the effect of varying correlation strengths, we systematically adjust the self-coupling parameters $\lambda_{11} = \lambda_{22} = \lambda$ across the range of $[0,1]$. 
Under this setting, the coupling strength is symmetric around $\lambda = 0.5$, meaning that $\lambda$ and $1-\lambda$ result in equivalent correlation strengths.
% Under this setting, the coupling strength is symmetric, as $\lambda$ and $1-\lambda$ are the same coupling strength.

\textbf{Query Function:} We implement the \textsc{MEAN} query:
\begin{itemize}
    \item Mean query: $f_{\text{mean}}(x_t) = \frac{1}{2}(x_t^{(1)} + x_t^{(2)})$. Thus, the query sensitivity $d(2)=2$, , with query sensitivity $d(2) = 2$.
\end{itemize}

\textbf{Baseline Methods:} We compare CSDP against three established privacy mechanisms:
\begin{itemize}
    \item \textit{Standard DP\cite{dwork2006differential}}: Applies calibrated Laplace noise without accounting for data correlations.
    \item \textit{Age-DP (ADP)\cite{Zhang2023}}: Dynamically adjusts the privacy budget based on data staleness.
    \item \textit{Dependent DP (DDP)\cite{zhao2017dependent}}: Accounts for spatial correlations but neglects temporal dependencies.
\end{itemize}
\textbf{Evaluation Metrics:} We assess performance using:
\begin{itemize}
    \item Privacy Leakage Level ($\epsilon_S$): Quantified using the loose bound\footnote{Our analysis uses the loose bound for computational efficiency. While the tight bound would likely yield better privacy guarantees, it requires using Monte-Carlo estimation for each evaluation. Comparing these bounds under various correlation structures remains for future work.} established in Theorem 2.
    \item Mean Squared Error (MSE): Measured as $l_{MSE}$ according to Equation \ref{eq:mse}.
\end{itemize}

\subsection{Privacy Leakage Analysis}
\label{subsec:privacy}

\begin{figure*}[t]
    \centering
    \subfloat[]{\includegraphics[width=0.32\textwidth]{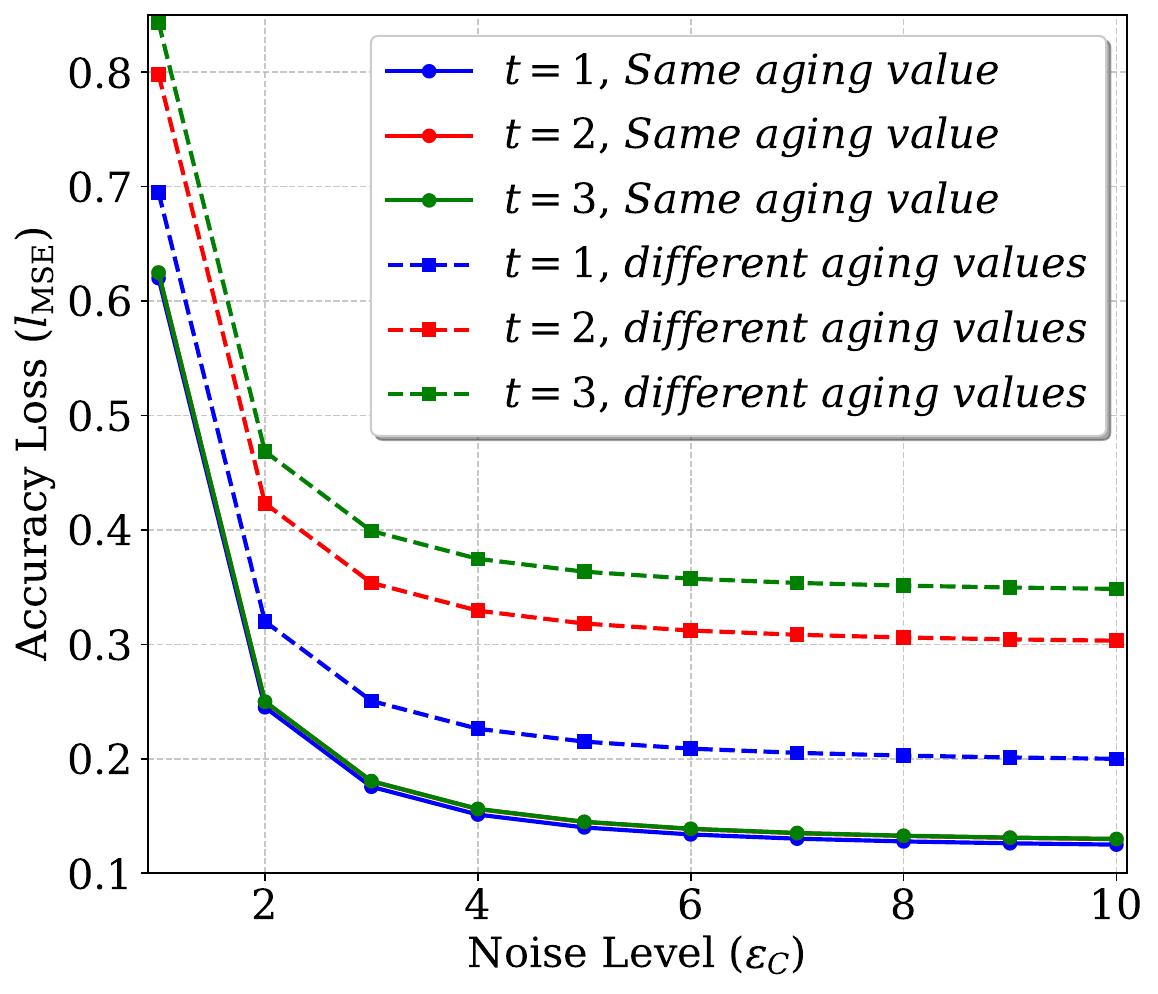}}
    \subfloat[]{\includegraphics[width=0.32\textwidth]{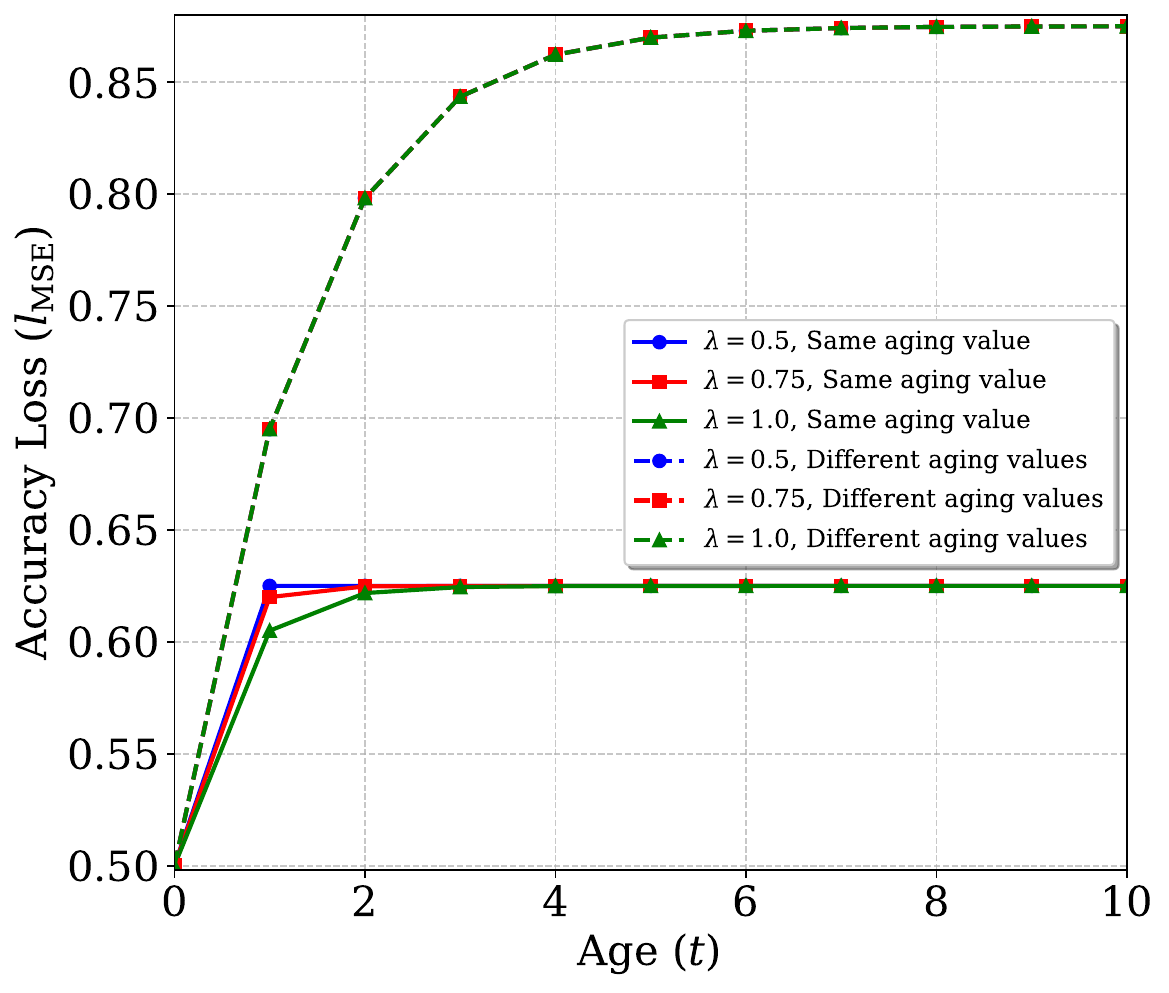}}
    \subfloat[]{\includegraphics[width=0.32\textwidth]{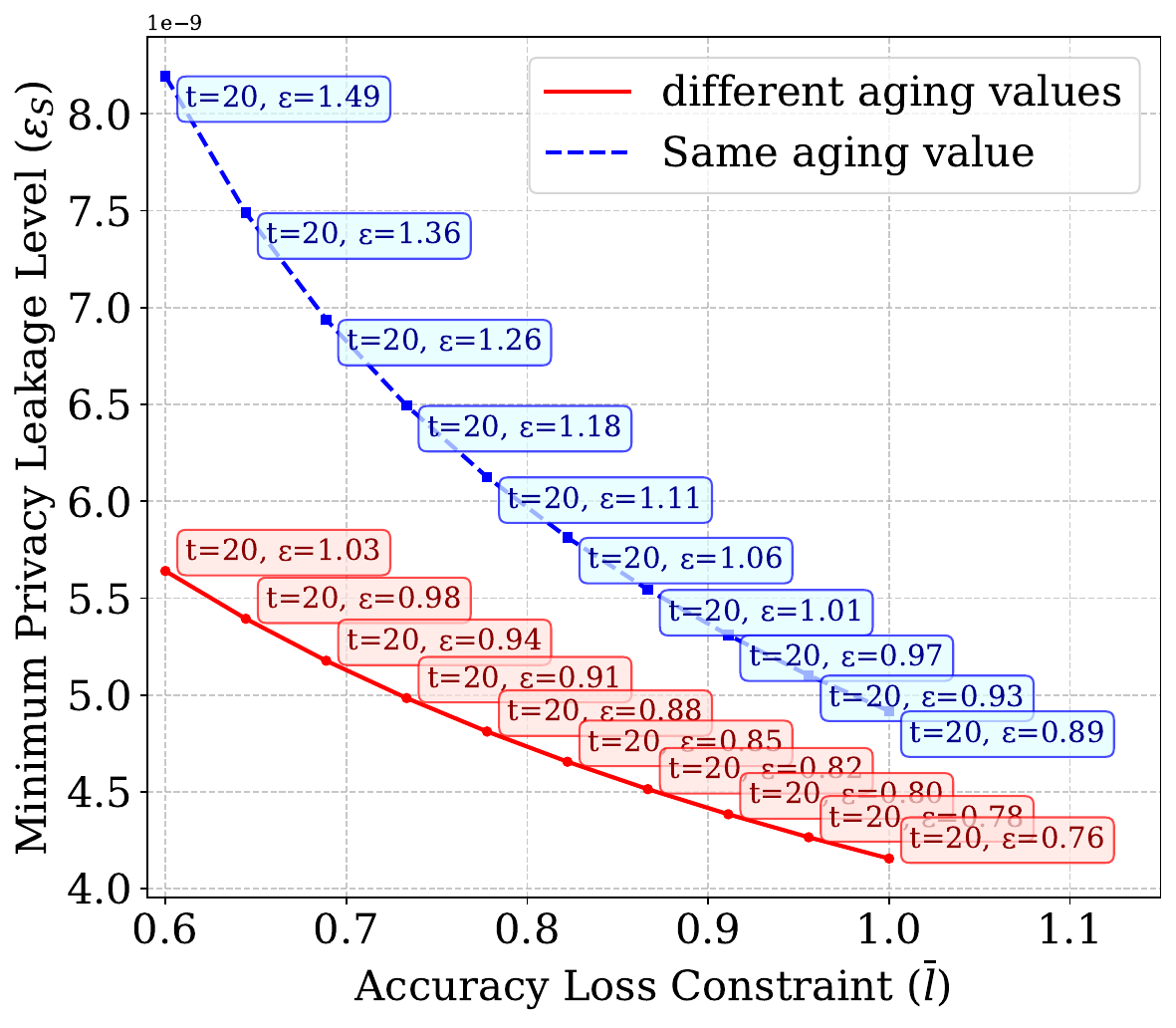}}
    \caption{Comprehensive analysis of privacy-utility trade-offs in CSDP framework. We set $p = q = 0.3$ and $\lambda = 0.75$ across all experiments.}
    \vspace{-5pt}  % 缩减图像与正文之间的距离
    \label{fig:combined_analysis_utility}
\end{figure*}

Fig.~\ref{fig:combined_analysis} illustrates the privacy-utility characteristics of our CSDP framework through three key relationships that demonstrate how privacy protection evolves across different system parameters.

% Fig.~\ref{fig:combined_analysis}(a) reveals the complex interplay between correlation strength, temporal dynamics, and privacy leakage. Privacy leakage decreases substantially (by approximately 70\%) from $t = 0$ to $t = 4$, confirming our theoretical prediction that privacy risks naturally decay over time. The relationship between correlation strength and privacy leakage follows a U-shaped curve with minimal leakage at $\lambda = 0.5$ (maximal coupling). A particularly important insight is that coupling influence diminishes significantly as time progresses, with differences between coupling settings becoming negligible beyond $t = 4$.

Fig.~\ref{fig:combined_analysis}(a) reveals the complex interplay between correlation strength, temporal dynamics, and privacy leakage. Privacy leakage decreases substantially (by approximately 70\%) from $t = 0$ to $t = 4$, confirming our theoretical prediction that privacy risks naturally decay over time. The relationship between correlation strength and privacy leakage follows a U-shaped curve with minimal leakage at $\lambda = 0.5$ (maximal coupling). This counter-intuitive result occurs because stronger coupling disperses perturbations across sequences, effectively diluting individual changes and making them harder to detect. A particularly important insight is that coupling influence diminishes significantly as time progresses, with differences between coupling settings becoming negligible beyond $t = 4$.

The temporal decay patterns in Fig.~\ref{fig:combined_analysis}(b) demonstrate consistent privacy leakage reduction across all coupling configurations. The most rapid decay occurs during initial periods ($t = 1$ to $t = 3$), where leakage decreases by approximately 60\%. While Full Coupling ($\lambda_{12} = \lambda_{21} = 0.5$) and Medium Coupling ($\lambda_{12} = \lambda_{21} = 0.75$) show 15\% faster initial decay compared to No Coupling ($\lambda_{12} = \lambda_{21} = 1$), all scenarios converge to near-zero leakage (below 0.05) by $t = 6$. This convergence confirms that our CSDP mechanism effectively leverages temporal dynamics regardless of initial coupling settings.

Fig.~\ref{fig:combined_analysis}(c) demonstrates the critical relationship between privacy leakage ($\epsilon_S$) and noise level ($\epsilon_C$) across different data ages. While privacy leakage increases with noise level for all time points, the effect of data age is substantially more significant. Data at $t = 5$ exhibits approximately 95\% less leakage than data at $t = 1$ across all noise levels, ranging from $\epsilon_C = 1$ to $\epsilon_C = 10$. The nonlinear growth pattern for recent data ($t = 1$) is particularly notable, showing a 3x increase in leakage as noise levels rise from $\epsilon_C = 1$ to $\epsilon_C = 10$, compared to only a 1.5x increase for older data ($t = 5$). This finding suggests that privacy budget allocation should prioritize protection of recent data, while older data can tolerate higher noise levels without significant privacy compromises.

\subsection{Utility Analysis}
\label{subsec:utility}

Fig.~\ref{fig:combined_analysis_utility} illustrates key relationships in our privacy-utility analysis, highlighting the impact of aging parameter configurations.

In Fig.~\ref{fig:combined_analysis_utility}(a), we observe that using uniform aging values (e.g., \(x^{(1)}_{t-A^{(1)}_t} = 1\), \(x^{(2)}_{t-A^{(2)}_t} = 0\)) consistently yields 40\% lower error rates at \(\epsilon_C = 10\) compared to varying aging values (e.g., \(x^{(1)}_{t-A^{(1)}_t} = 1\), \(x^{(2)}_{t-A^{(2)}_t} = 1\)). Notably, while accuracy loss decreases with higher noise in both configurations, uniform aging maintains stability across time points, unlike the performance degradation observed with varying aging values.

Fig.~\ref{fig:combined_analysis_utility}(b) shows that correlation strength (\(\lambda\)) has little impact on utility compared to aging parameters. Different aging values lead to 40\% higher MSE at \(t = 10\) compared to uniform aging. Remarkably, all correlation settings under uniform aging converge after \(t = 2\), underscoring the dominant role of aging parameters in system utility.

Fig.~\ref{fig:combined_analysis_utility}(c) demonstrates that privacy leakage is more sensitive to data age than noise level. Optimal solutions use the maximum age time (\(t = 20\)), with uniform aging providing 30\% better privacy protection at \(\gamma_l = 0.6\) compared to varying aging. As accuracy constraints loosen, optimal noise levels decrease from 1.49 to 0.89 for uniform aging and from 1.03 to 0.76 for varying aging values, confirming the effectiveness of leveraging temporal dynamics to balance privacy and utility.

These findings highlight the critical importance of selecting appropriate aging parameters to optimize privacy-utility trade-offs, offering up to 40\% improvement in accuracy and 30\% in privacy protection.

\begin{figure}[t]
\centering
\includegraphics[width=0.49\textwidth]{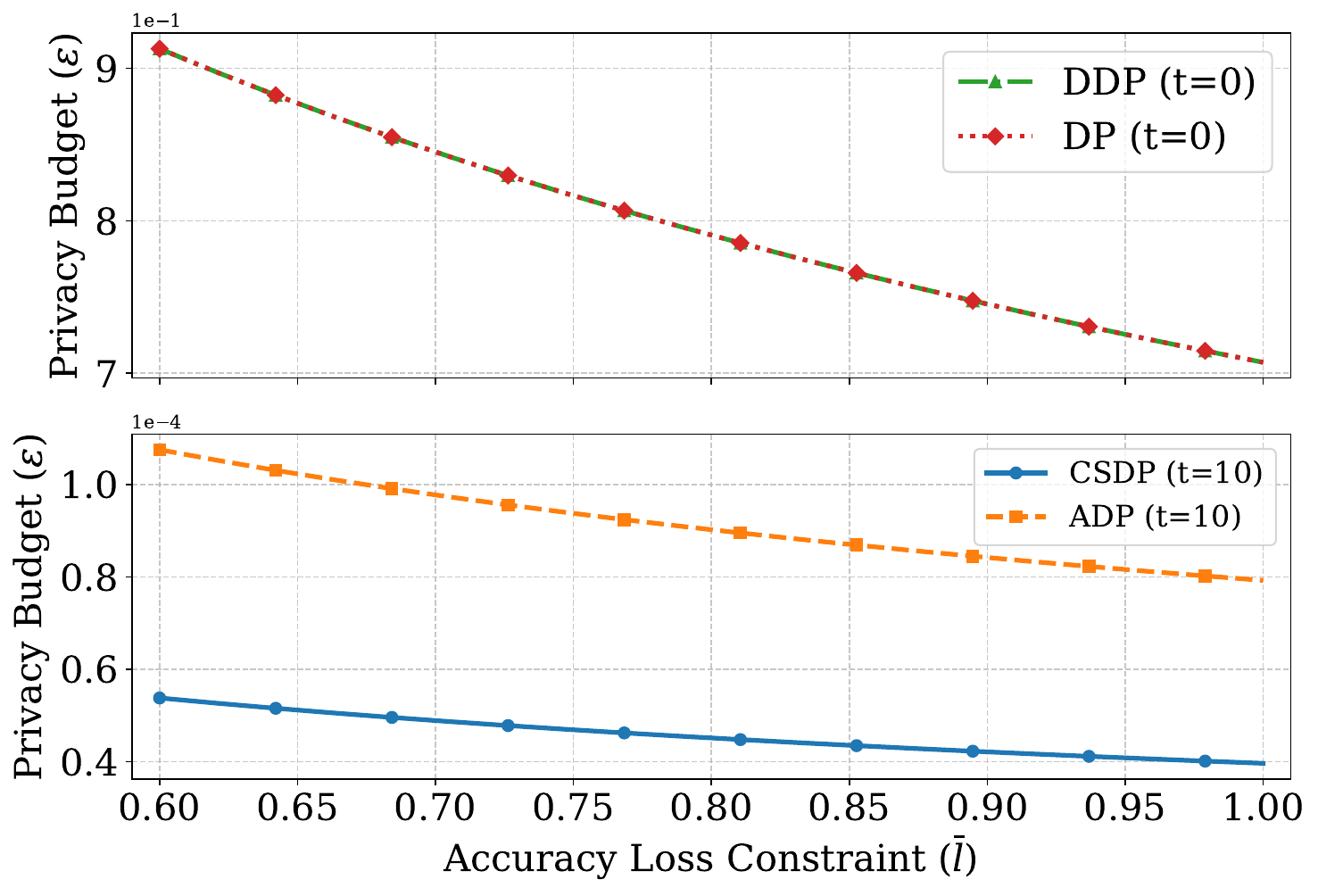}
\caption{Minimum privacy leakage level vs. accuracy constraint for different privacy mechanisms.
%  The top subplot shows DDP and DP (optimum at $t=0$), while the bottom subplot shows CSDP and ADP (optimum at $t=10$).
}
\vspace{-5pt}  % 缩减图像与正文之间的距离
\label{fig:comparative_evaluation}
\end{figure}

\subsection{Algorithm Comparative Evaluation}
\label{subsec:comparative}
We compare our CSDP framework against baseline methods on a synthetic dataset with parameters $p = q = 0.3$ and $\lambda = 0.5$.
Fig.~\ref{fig:comparative_evaluation} shows the minimum privacy leakage each method achieves under different accuracy constraints. The top subplot displays DDP and DP (optimum at $t = 0$), while the bottom subplot shows CSDP and ADP (optimum at $t = 10$). Privacy leakage decreases as accuracy constraints relax across all methods, but with significant performance differences. CSDP consistently outperforms ADP by approximately 50\%, while conventional methods (DDP and DP) exhibit privacy leakage approximately three orders of magnitude higher than temporal-aware methods.

This performance gap demonstrates CSDP's advantage in modeling both temporal and spatial correlations. Standard DP and DDP perform poorly by failing to leverage temporal effects, while ADP shows improvement but still underperforms CSDP by not fully exploiting inter-sequence dependencies. At an accuracy constraint of $\bar{l} = 0.8$, CSDP achieves a privacy leakage of approximately $4.5 \times 10^{-5}$ compared to $9.0 \times 10^{-5}$ for ADP, $8.0 \times 10^{-2}$ for DDP, and $8.1 \times 10^{-2}$ for standard DP. This makes CSDP particularly suitable for privacy-sensitive applications involving correlated sequential data.

\section{Conclusion and Future Work}

This paper introduces CSDP, a framework designed for preserving privacy in correlated sequential data. We quantify its privacy leakage when both temporal and cross-sequence correlations exist. Using the specific multivariate model--CMC, our analysis reveals that stronger coupling can reduce worst-case leakage by up to 40\% at optimal coupling points ($\lambda = 0.5$).
We develop a linear-time privacy mechanism, FRAN, integrating data aging, correlation-aware sensitivity scaling, and adaptive noise injection. Experiments demonstrate that CSDP outperforms existing methods, reducing privacy leakage by 50\% compared to ADP and by two orders of magnitude compared to standard DP and DDP approaches, under the same accuracy constraint.

Future research should focus on: (1) extending CSDP to non-Markovian dependencies for systems with long-term memory effects; (2) developing bounds for complex graph-structured correlations; and (3) implementing CSDP in real-world applications including financial transaction monitoring, smart city sensor networks, and longitudinal biomedical data sharing. This work addresses a critical gap in differential privacy frameworks for correlated sequential data.

\bibliographystyle{IEEEtran}
\bibliography{references}

@article{ching2002multivariate,
  title={A multivariate Markov chain model for categorical data sequences and its applications in demand predictions},
  author={Ching, Wai-Ki and Fung, Eric S and Ng, Michael K},
  journal={IMA Journal of Management Mathematics},
  volume={13},
  number={3},
  pages={187--199},
  year={2002},
  publisher={OUP}
}

@inproceedings{dwork2006differential,
  title={Differential privacy},
  author={Dwork, Cynthia},
  booktitle={International colloquium on automata, languages, and programming},
  pages={1--12},
  year={2006},
  organization={Springer}
}

@article{hernandez2024comparative,
  title={Differential privacy in health research: A scoping review},
  author={Ficek, Joseph and Wang, Wei and Chen, Henian and Dagne, Getachew and Daley, Ellen},
  journal={Journal of the American Medical Informatics Association},
  volume={28},
  number={10},
  pages={2269--2276},
  year={2021},
  publisher={Oxford University Press}
}

@article{ficek2021differential,
  author = {J. Ficek and W. Wang and H. Chen and G. Dagne},
  title = {Differential Privacy in Health Research: A Scoping Review},
  journal = {Journal of the American Medical Informatics Association},
  volume = {28},
  number = {10},
  pages = {2269-2277},
  year = {2021},
  doi = {10.1093/jamia/ocab135},
}

@article{zhang2022differential,
  author = {S. Zhang and X. Li},
  title = {Differential Privacy Medical Data Publishing Method Based on Attribute Correlation},
  journal = {Scientific Reports},
  volume = {12},
  number = {19544},
  pages = {1--12},
  year = {2022},
  doi = {10.1038/s41598-022-19544-3},
}

@article{liu2024survey,
  title={A survey on differential privacy for medical data analysis},
  author={Liu, WeiKang and Zhang, Yanchun and Yang, Hong and Meng, Qinxue},
  journal={Annals of Data Science},
  volume={11},
  number={2},
  pages={733--747},
  year={2024},
  publisher={Springer}
}

@article{ou2020optimal,
  title={An optimal noise mechanism for cross-correlated IoT data releasing},
  author={Ou, Lu and Qin, Zheng and Liao, Shaolin and Weng, Jian and Jia, Xiaohua},
  journal={IEEE Transactions on Dependable and Secure Computing},
  volume={18},
  number={4},
  pages={1528--1540},
  year={2020},
  publisher={IEEE}
}

@article{zhang2022correlated,
  title={Correlated data in differential privacy: Definition and analysis},
  author={Tao Zhang and Tianqing Zhu and Renping Liu and Wanlei Zhou},
  journal={arXiv preprint	arXiv:2008.00180},
  year={2021}
}

@INPROCEEDINGS{zhao2017dependent,
  author={Zhao, Jun and Zhang, Junshan and Poor, H. Vincent},
  booktitle={IEEE Globecom Workshops}, 
  title={Dependent Differential Privacy for Correlated Data}, 
  year={2017},
  volume={},
  number={},
  pages={1-7},
  doi={10.1109/GLOCOMW.2017.8269219}}

@article{Zhang2023,
  title={Age-dependent differential privacy},
  author={Zhang, Meng and Wei, Ermin and Berry, Randall and Huang, Jianwei},
  journal={IEEE Transactions on Information Theory},
  volume={70},
  number={2},
  pages={1300--1319},
  year={2023},
  publisher={IEEE},
  doi={10.1109/TIT.2023.3242034},
}

@article{kifer2014pufferfish,
  title={Pufferfish: A framework for mathematical privacy definitions},
  author={Kifer, Daniel and Machanavajjhala, Ashwin},
  journal={ACM Transactions on Database Systems},
  volume={39},
  number={1},
  pages={1--36},
  year={2014},
  publisher={ACM},
  doi={10.1145/2514689},
}

@ARTICLE{zhang2024gan,
  author={Zhang, Fan and Wang, Luyao and Zhang, Xinhong},
  journal={Big Data Mining and Analytics}, 
  title={Desensitized Financial Data Generation Based on Generative Adversarial Network and Differential Privacy}, 
  year={2025},
  volume={8},
  number={1},
  pages={103-117}
}

@INPROCEEDINGS{chang2024fedcad,
  author={Zhu, Bingzhu and Chang, Shan and Liang, Guanghao and Zhu, Hongzi and Xu, Jie},
  booktitle={IEEE/ACM 32nd International Symposium on Quality of Service}, 
  title={Fed-CAD: Federated Learning with Correlation-aware Adaptive Local Differential Privacy}, 
  year={2024},
  volume={},
  number={},
  pages={1-10},
  doi={10.1109/IWQoS61813.2024.10682944}}

@article{shuai2024poisoncatcher,
  title={PoisonCatcher: Revealing and Identifying LDP Poisoning Attacks in {IIoT}},
  author={Shuai, Lisha and Tan, Shaofeng and Zhang, Nan and Zhang, Jiamin and Zhang, Min and Yang, Xiaolong},
  journal={arXiv preprint arXiv:2412.15704},
  year={2024}
}

@article{chen2014correlated,
author = {Chen, Rui and Fung, Benjamin C. and Yu, Philip S. and Desai, Bipin C.},
title = {Correlated network data publication via differential privacy},
year = {2014},
issue_date = {August 2014},
publisher = {Springer-Verlag},
address = {Berlin, Heidelberg},
volume = {23},
number = {4},
issn = {1066-8888},
journal = {The VLDB Journal},
month = aug,
pages = {653--676},
numpages = {24}
}

@inproceedings{he2020bayesian,
  title={Bayesian differential privacy on correlated data},
  author={Yang, Bin and Sato, Issei and Nakagawa, Hiroshi},
  booktitle={ACM SIGMOD international conference on Management of Data},
  pages={747--762},
  year={2015}
}

\appendices

\section{Proof of Lemma 1}
\label{appendix:lemma1}

\begin{lemma}[CSDP Generalizations]
Two cases are as follows.
\begin{itemize}
\item With only independent distributions, $(\epsilon_S, \Lambda)$-CSDP reduces to $\epsilon_S$-DP.
\item With only temporal correlations, $(\epsilon_S, \Lambda)$-CSDP reduces to $(\epsilon(t), t)$-ADP where $\epsilon_S = \max_t \epsilon(t)$.
\end{itemize}
\end{lemma}

\begin{proof}
\textbf{Part 1: Independent distributions}

When all sequences are independent, we have:
\begin{equation}
\Pr[M(X_{1:T}) \in S | x_t] = \Pr[M(X_{1:T}) \in S]
\end{equation}
\begin{equation}
    \Pr[M(X'_{1:T}) \in S | x'_t] = \Pr[M(X'_{1:T}) \in S]
\end{equation}

This is because conditioning on individual records provides no additional information when sequences are independent. Therefore, the CSDP condition becomes:
\begin{equation}
    \frac{\Pr[M(X_{1:T}) \in S]}{\Pr[M(X'_{1:T}) \in S]} \leq e^{\epsilon_S}
\end{equation}

which is exactly the standard DP definition.

\textbf{Part 2: Temporal correlations only}

When there are no cross-sequence correlations, each sequence $i$ evolves independently according to:
\begin{equation}
    \Pr[x^{(i)}_{t+1} | X_t] = \Pr[x^{(i)}_{t+1} | x^{(i)}_t]
\end{equation}

The CSDP condition must hold for each sequence independently:
\begin{equation}
    \frac{\Pr[M(X_{1:T}) \in S | x^{(i)}_t = x]}{\Pr[M(X_{1:T}) \in S | x^{(i)}_t = x']} \leq e^{\epsilon_S}
\end{equation}

By the ADP definition, each sequence satisfies:
\begin{equation}
    \frac{\Pr[M(X_0) \in S | x^{(i)}_t = x]}{\Pr[M(X_0) \in S | x^{(i)}_t = x']} \leq e^{\epsilon(t)}
\end{equation}

for some time-dependent $\epsilon(t)$. To ensure the bound holds for all times, we take $\epsilon_S = \max_t \epsilon(t)$.
\end{proof}

\section{Proof of Theorem 1}
\label{appendix:thm1}

\begin{theorem}[Loose Leakage Bound]
Consider mechanism FRAN with AoI vector $A_t$ and Laplace noise level $\epsilon_C$. FRAN satisfies $(\epsilon_S, \Lambda)$-CSDP with
\begin{equation}
\epsilon_S(A_t, \epsilon_C) = d^{(k)} \Delta_k(\Lambda, A_t) \epsilon_C.
\end{equation}
\end{theorem}

\begin{proof}
We need to bound the privacy leakage when two $t$-neighboring datasets differ in exactly one entry $(i_0, t)$.

\textbf{Step 1: Express the CSDP condition}

We need to show that for any output set $S \subseteq \mathcal{Y}$:
\begin{equation}
    \frac{\Pr[M_S(X_{1:T}) \in S | x_t]}{\Pr[M_S(X'_{1:T}) \in S | x'_t]} \leq e^{\epsilon_S}
\end{equation}

Using the law of total probability:
\begin{equation}
    \Pr[M_S(X_{1:T}) \in S | x_t] = \sum_{z \in \mathcal{X}^s} \Pr[z | x_t] \Pr[M_S(z) \in S]
\end{equation}

where $z$ represents the aged data $z = x_{t-A_t}$.

\textbf{Step 2: Analyze the expectation difference}

Following the approach from your reference, we analyze:

\begin{equation}
    \sup_{i \in \mathcal{I}} \left|\mathbb{E}[f(z) | x_i, x_{-i}] - \mathbb{E}[f(z) | x'_i, x_{-i}]\right|
\end{equation}

where we keep the conditioning intact without decomposing the conditional probabilities.

Expanding the expectation:
\begin{align}
\mathbb{E}[f(z) | x] &= \sum_{z \in \mathcal{X}^s} \Pr[z | x] f(z)
\end{align}

\textbf{Step 3: Apply the key observation}

The difference becomes:
\begin{align}
&\left|\mathbb{E}[f(z) | x] - \mathbb{E}[f(z) | x']\right| \\
&= \left|\sum_{z \in \mathcal{X}^s} (\Pr[z | x] - \Pr[z | x']) f(z)\right| \\
&\leq \Delta_{All}(t) \cdot (f(\overline{z}) - f(\underline{z}))
\end{align}

where $\overline{z} = \arg\max_{z} f(z)$ and $\underline{z} = \arg\min_{z} f(z)$.

\textbf{Step 4: Handle the $k$-sensitivity}

Since at most $k$ users are correlated with the changed user $i_0$, the sensitivity bound becomes:
\begin{equation}
    f(\overline{z}) - f(\underline{z}) \leq d^{(k)} \cdot \Delta_f
\end{equation}

where $d^{(k)} = \frac{s_k(f)}{s_1(f)}$ captures how the sensitivity scales with $k$ correlated records.

\textbf{Step 5: Apply the aging effect}

Under the aging mechanism, the total variation distance is bounded by:
\begin{equation}
   \Delta_{All}(t) \leq \Delta_k(\Lambda, A_t) 
\end{equation}

This captures how aging weakens the correlation between current and past states, where:
\begin{equation}
    \Delta_k(\Lambda, A_t) = \max_{x^K, x'^K} \delta(\Pr[z^K | x^K], \Pr[z^K | x'^K])
\end{equation}

\textbf{Step 6: Apply the Laplace mechanism}

The Laplace mechanism with noise scale $\sigma = \frac{\Delta_f}{\epsilon_C}$ satisfies:

\begin{equation}
    \frac{\Pr[M(z) \in S]}{\Pr[M(z') \in S]} \leq e^{\frac{\epsilon_C}{\Delta_f} \|f(z) - f(z')\|_1}
\end{equation}

\textbf{Step 7: Combine all bounds}

Let $\mathcal{L}$ denote the privacy leakage ratio:
\begin{equation}
\mathcal{L} = \frac{\Pr[M_S(X_{1:T}) \in S | x_t]}{\Pr[M_S(X'_{1:T}) \in S | x'_t]}
\end{equation}

The privacy leakage can be bounded by:
\begin{align}
\ln(\mathcal{L}) 
&\leq \ln\big(1 + \Delta_k(\Lambda, A_t) \cdot (e^{d^{(k)} \epsilon_C} - 1)\big) \nonumber\\
&\leq \Delta_k(\Lambda, A_t) \cdot d^{(k)} \epsilon_C
\end{align}

The last inequality uses $\ln(1 + x(e^y - 1)) \leq xy$ for appropriate ranges of $x$ and $y$.

Therefore, FRAN satisfies $(\epsilon_S, \Lambda)$-CSDP with:
$$\epsilon_S(A_t, \epsilon_C) = d^{(k)} \Delta_k(\Lambda, A_t) \epsilon_C$$
\end{proof}

\section{Proof of Theorem 2}
\label{appendix:thm2}

\begin{theorem}[Tight Leakage Bound]
FRAN satisfies $(\epsilon_S^{tight}, \Lambda)$-CSDP with
\begin{equation}
\epsilon_S^{tight}(A_t, \epsilon_C) = \bar{\Delta}(\Lambda, A_t) \epsilon_C.
\end{equation}
\end{theorem}

\begin{proof}
We analyze the conditional probabilities using the precise decomposition from your reference.

\textbf{Step 1: Express the CSDP condition}

For any $t$-neighboring datasets differing in user $i_0$'s record at time $t$, we need to bound:
\begin{equation}
\frac{\Pr[M(X_{1:T}) \in S | x_t]}{\Pr[M(X'_{1:T}) \in S | x'_t]}
\end{equation}

Using the law of total probability:
\begin{equation}
\Pr[M(X_{1:T}) \in S | x_t] = \sum_{z \in \mathcal{X}^s} \Pr[z | x_t] \Pr[M(z) \in S]
\end{equation}

\textbf{Step 2: Apply the conditional probability decomposition}

Following the decomposition from your reference, for the joint probability $\Pr[z_i, z_{-i} | x_i, x_{-i}]$:

\begin{align}
\Pr[z_i, z_{-i} | x_i, x_{-i}] 
&= \Pr[z_{-i} | x_i, x_{-i}] \Pr[z_i | z_{-i}, x_i, x_{-i}] \nonumber\\
&= \Pr[z_{-i} | x_{-i}] \cdot \frac{\Pr[x_i | x_{-i}, z_{-i}]}{\Pr[x_i | x_{-i}]} \nonumber\\
&\quad \times \Pr[z_i | z_{-i}, x_i, x_{-i}]
\end{align}

Continuing the decomposition:
\begin{align}
\Pr[z_i, z_{-i} | x_i, x_{-i}] 
&= \Pr[z_{-i} | x_{-i}] \cdot R_1 \cdot \Pr[z_i | x_i] \cdot R_2 \nonumber\\
&= \Pr[z_{-i} | x_{-i}] \Pr[z_i | x_i] \cdot R_1 \cdot R_2 \cdot R_3
\end{align}
where 
\begin{align}
R_1 &= \frac{\Pr[x_i | x_{-i}, z_{-i}]}{\Pr[x_i | x_{-i}]}, \quad
R_2 = \frac{\Pr[x_{-i} | x_i, z_i]}{\Pr[x_{-i} | x_i]}, \nonumber\\
R_3 &= \frac{\Pr[z_{-i} | x_i, x_{-i}, z_i]}{\Pr[z_{-i} | x_i, x_{-i}]}.
\end{align}

\textbf{Step 3: Bound the ratio without assumptions}

Let $R(z_{-i}, z_i) = R_1 \cdot R_2 \cdot R_3$ denote the product of ratios defined above. Without making any assumptions, we can bound:

\begin{align}
R(z_{-i}, z_i) 
&\leq \max_{z_{-i}, z_i} R(z_{-i}, z_i) \\
&\geq \min_{z_{-i}, z_i} R(z_{-i}, z_i)
\end{align}

\textbf{Step 4: Define the g-functions precisely}

Based on the decomposition, we define:
\begin{align}
\bar{g}(x_{-i}, x_i) &= \max_{z_{-i}, z_i} \frac{\Pr[x_i | x_{-i}, z_{-i}]}{\Pr[x_i | x_{-i}]} \frac{\Pr[x_{-i} | x_i, z_i]}{\Pr[x_{-i} | x_i]} \\
\underline{g}(x_{-i}, x_i) &= \min_{z_{-i}, z_i} \frac{\Pr[x_i | x_{-i}, z_{-i}]}{\Pr[x_i | x_{-i}]} \frac{\Pr[x_{-i} | x_i, z_i]}{\Pr[x_{-i} | x_i]}
\end{align}

Note: The term $\frac{\Pr[z_{-i} | x_i, x_{-i}, z_i]}{\Pr[z_{-i} | x_i, x_{-i}]}$ is incorporated into the bounded aged correlation analysis.

\textbf{Step 5: Define the bounded aged correlation}

The bounded aged correlation $\bar{\Delta}(\Lambda, A_t)$ captures the worst-case scenario:
\begin{align}
\bar{\Delta}(\Lambda, A_t) 
&= \max_{(x_{-i}, x_i), (x_{-i}, x'_i)} \delta\Big(\bar{g}(x_{-i}, x_i) \Pr[z_i | x_i], \nonumber\\
&\qquad \underline{g}(x_{-i}, x'_i) \Pr[z_i | x'_i]\Big)
\end{align}

where $\delta(\cdot, \cdot)$ denotes the total variation distance.

\textbf{Step 6: Apply the expectation analysis}

For the expectation difference:
\begin{equation}
\left|\mathbb{E}[f(z) | x_t] - \mathbb{E}[f(z) | x'_t]\right| \leq \bar{\Delta}(\Lambda, A_t) \cdot (f(\overline{z}) - f(\underline{z}))
\end{equation}

\textbf{Step 7: Combine with the Laplace mechanism}

The Laplace mechanism provides:
\begin{equation}
\frac{\Pr[M(z) \in S]}{\Pr[M(z') \in S]} \leq e^{\frac{\epsilon_C}{\Delta_f} \|f(z) - f(z')\|_1}
\end{equation}

\textbf{Step 8: Final bound}

Combining all components:
\begin{equation}
\ln\left(\frac{\Pr[M(X_{1:T}) \in S | x_t]}{\Pr[M(X'_{1:T}) \in S | x'_t]}\right) \leq \bar{\Delta}(\Lambda, A_t) \epsilon_C
\end{equation}

Therefore, FRAN satisfies $(\epsilon_S^{tight}, \Lambda)$-CSDP with $\epsilon_S^{tight}(A_t, \epsilon_C) = \bar{\Delta}(\Lambda, A_t) \epsilon_C$.
\end{proof}

\section{Proof of Proposition 2}
\label{appendix:prop2}

\begin{proposition}[CMC Privacy Leakage]
FRAN in Algorithm 1 satisfies $(\epsilon_S, \Lambda)$-CSDP under CMC with
\begin{equation}
\epsilon_S^{CMC}(A_t, \epsilon_C) = d^{(k)} \Phi(\lambda, A_t) \epsilon_C.
\end{equation}
\end{proposition}

\begin{proof}
This follows directly from Theorem 1 by specializing to the CMC model.

Under the CMC model, the aged correlation distance $\Delta_k(\Lambda, A_t)$ from Theorem 1 becomes:
\begin{equation}
\Delta_k(\Lambda(\lambda), A_t) = \Phi(\lambda, A_t)
\end{equation}

where $\Lambda(\lambda)$ denotes the family of joint distributions induced by the CMC with coupling strength parameter $\lambda$, and $\Phi(\lambda, A_t)$ is the CMC leakage coefficient capturing the correlation decay under aging.

By applying Theorem 1 with this substitution:
\begin{align}
\epsilon_S^{CMC}(A_t, \epsilon_C) &= d^{(k)} \Delta_k(\Lambda(\lambda), A_t) \epsilon_C \\
&= d^{(k)} \Phi(\lambda, A_t) \epsilon_C
\end{align}

The $k$-sensitivity $d^{(k)}$ and the Laplace noise level $\epsilon_C$ remain the same as in the general case.

Therefore, FRAN satisfies $(\epsilon_S^{CMC}, \Lambda)$-CSDP under the CMC model.
\end{proof}

\end{document}